\documentclass[a4paper,preprint, numbers,sort&compress,3p,11pt]{elsarticle}
\usepackage{orcidlink}

\usepackage[american]{babel} 
\usepackage[utf8]{inputenc}
\usepackage[T1]{fontenc}

\usepackage{amsmath}
\usepackage{amsthm} 
\usepackage{amssymb} 
\usepackage{mathtools}
\usepackage{mathrsfs}

\usepackage{graphicx}
\usepackage{xcolor}
\usepackage{fca}
\usepackage{marvosym}
\usepackage[ruled,vlined,linesnumbered]{algorithm2e}
\usepackage{cleveref}
\usepackage{enumitem}
\usepackage{listings}
\usepackage{booktabs}
\usepackage{tikz}
\usepackage{tikzscale}

\usepackage{MnSymbol}

\usepackage{pict2e}

\newcommand{\todo}[1]{{\bf\color{red} Todo: #1}}
\newcommand{\quest}[1]{}

\newcommand{\longversion}[1]{}

\newcommand{\K}{\mathbb{K}}

\newcommand{\N}{\mathbb{N}}
\newcommand{\UL}{\underline{L}}
\newcommand{\US}{\underline{S}}
\newcommand{\B}[1]{\mathfrak{B}(#1)}
\newcommand{\BB}{\mathfrak{B}}
\newcommand{\UBB}{\underline{\mathfrak{B}}}

\let\cref\Cref

\newif\ifhideproofs

\ifhideproofs
\usepackage{environ}
\NewEnviron{hide}{}

\fi

\newtheorem{theorem}{Theorem}[section] 
\newtheorem{corollary}[theorem]{Corollary}
\newtheorem{lemma}[theorem]{Lemma}
\newtheorem{proposition}[theorem]{Proposition} 

\newtheorem{definition}[theorem]{Definition}
\newtheorem{example}[theorem]{Example}

\begin{document}
\begin{frontmatter}    
	
	\title{Factorizing Lattices by Interval Relations}
	
	\author{Maren Koyda\corref{cor1} \orcidlink{0000-0002-8903-6960}} \ead{koyda@cs.uni-kassel.de}
	
	\author{Gerd Stumme\corref{cor2} \orcidlink{0000-0002-0570-7908}} \ead{stumme@cs.uni-kassel.de}
	\cortext[cor1]{Corresponding author}
	\address{Knowledge and Data Engineering Group,\\[0.5ex]
		Research Center for Information System Design,\\[0.5ex]
		University of Kassel, Wilhelmshöher Allee 73, D-34121 Kassel,
		Germany}
	
\begin{abstract}
	This work investigates the factorization of finite lattices to implode selected intervals while preserving the remaining order structure.
	We examine how complete congruence relations and complete tolerance relations can be utilized for this purpose and answer the question of finding the finest of those relations to implode a given interval in the generated factor lattice.
	To overcome the limitations of the factorization based on those relations, we introduce a new lattice factorization that enables the imploding of selected disjoint intervals of a finite lattice.
	To this end, we propose an \emph{interval relation} that generates this factorization.
	To obtain lattices rather than arbitrary ordered sets, we restrict this approach to so-called \emph{pure intervals}.
	For our study, we will make use of methods from Formal Concept Analysis (FCA).
	We will also provide a new FCA construction by introducing
	the \emph{enrichment} of an incidence relation by a set of intervals in a formal context, to investigate the approach for lattice-generating interval relations on the context side.
\end{abstract}

	
\end{frontmatter}

\section{Introduction}
\label{sec:introduction}

A lattice $\UL$ consists of several intervals 
(which means subsets $[u,v]\coloneqq\{c\in\UL\mid u\le c, c\le v\}$).
Such an interval is often considered as a unity.
Therefore, a transformation of the lattice that implodes such an interval (or a set of intervals) is a suitable way to obtain a more compact representation of the original structure.
In this work, we investigate different ways to factorize a lattice to ``implode'' selected intervals,
meaning to condense an interval into a single element in the following manner:

\begin{definition}
	\label{def:implosion}
	An \emph{implosion} of a set of disjoint intervals $\US_1,\dots,\US_k$ in an ordered set $(P,\le_P)$ is a surjective, order-preserving map $f\colon (P,\le_P)\rightarrow (Q\le_Q)$ on an ordered set $Q$ so that $|f(\US_i)|=1$ for all $i\in\{1,\dots,k\}$.
\end{definition}


We pursue two -- partially conflicting -- goals for such an implosion.
The first goal is that the function $f$ is as compatible as possible with the lattice structure -- in the best case, f is a lattice homomorphism.
The second goal is that the remaining part of the lattice remains intact as much as possible -- in the best case, $f$ is injective in $\UL\setminus\US$.
As we will discuss in the sequel, there are different ways to realize an implosion of one  or more intervals with different trade-offs regarding our goals.

After recalling some definitions in~\cref{sec:foundations} and discussing the state of the art in~\cref{sec:related-work}, we will, in~\cref{sec:fact_substr}, examine how lattice factorizations based on complete congruence relations and complete tolerance relations can be utilized for imploding intervals in a lattice.
In particular, we answer the question of finding the finest of those relations 
that implodes a given interval in the generated factor lattice while preserving as many of the other elements of the lattice as possible.
While both of these approaches result in an infimum- and supremum-preserving factor lattice, the imploded intervals are often larger than just the selected interval.
It is even possible for the factor lattice to implode the whole lattice and thus to contain only one element.
To preserve all elements of the lattice except those in selected intervals we investigate then, in~\cref{sec:new_relation}, a new kind of factorization based on \emph{interval relations}.
\longversion{
}

We will use methods of Formal Concept Analysis (FCA)~\cite{fca-book} to study the different types of implosion.
In this field, the basic structures are \emph{formal contexts} which consist of a set of objects, a set of attributes, and an incidence relation, that represents which object \emph{has} which attribute.
The maximal sets of objects that share an identical maximal set of attributes are called the \emph{formal concepts} and can be ordered by the subset relation.
This determines a complete lattice, the \emph{concept lattice} corresponding to the formal context.
We introduce \emph{enrichments of the incidence relation by intervals} and show their one-to-one correspondence to interval relations. 
Those structures can be utilized to implode selected intervals in the lattice while preserving the original order relation.
By restricting the approach to \emph{pure intervals} we also ensure the lattice properties in the generated structure.

Since every finite lattice is isomorphic to a concept lattice, all statements can be translated to finite lattices in general.
For readability reasons we often omit the word ``finite'' in the following.
However, \textbf{all stetements are about finite structures only}.´


\section{Foundations}
\label{sec:foundations}

In the following, we recall basic notions from FCA and give notations that are used in this work.
For a more detailed introduction, we refer the reader to \cite{fca-book}.
A \emph{formal context} $\K\coloneqq(G,M,I)$ consists of 
a set $G$ whose elements are called \emph{objects}, a set $M$ whose elements are called \emph{attributes} and an \emph{incidence relation} $I\subseteq G\times M$. 
In this work, we assume that $G$ and $M$ are both finite.
Two operations, called \emph{derivations}, are defined on the power set of the objects and the power set of the attributes as follows:
$
\cdot'\colon\mathcal{P}(G)\to\mathcal{P}(M),~A\mapsto A'\coloneqq
\{m\in M\mid \forall g\in A\colon (g,m)\in I\}$ and $ 
\cdot'\colon\mathcal{P}(M)\to\mathcal{P}(G),~ B\mapsto B'\coloneqq\{g\in
G\mid \forall m\in B\colon (g,m)\in I\}
$. 
Instead of $A'$ we also write $A^I$ to point out the used incidence relation.
A pair $c=(A,B)$ consisting of an object subset $A\subseteq G$ and an attribute subset $B\subseteq M$ satisfying $A'=B$ and $B'=A$ is called a \emph{formal concept} of the context $(G,M,I)$.
$A$ is called the \emph{extent} and $B$ is called the \emph{intent} of $c$.
A concept of the form $c=(g'',g')$ ($c=(m',m'')$) is \emph{generated} by $g\in G$ (by $m\in M$).
\longversion{Further a \emph{minimal object generator} of a concept $c=(A,B)$ is the set $O\subseteq G$ with $O''=A$ so that $P''\not=A$ holds for every proper subsets $P\subsetneq O$.
The \emph{minimal attribute generator} of $c$ is defined in an analogously.
The set of all minimal object generators of $c$ is denoted by $minG_{obj}(c)$. 
The set $minG_{att}(c)$ stands for the set of all minimal attribute generators.}

On the set of all concepts, $\BB(\K)$, of a formal context $\K$, the order defined by $(A_1,B_1)\leq(A_2,B_2)$ iff $A_1\subseteq A_2$ determines the \emph{concept lattice} $\underline{\mathfrak{B}}(\mathbb{K})\coloneqq(\mathfrak{B}(\mathbb{K}),\leq)$.
For a lattice $\underline{L}$ we call the formal context $(\UL,\UL,\leq)$ the \emph{generic formal context of $\underline{L}$}. 
It holds that $\UL\cong\UBB(\UL,\UL,\leq)$ where
every concept is generated by a single object (a single attribute).
%

For the structural investigation of a formal context $\mathbb{K}=(G,M,I)$ we can \emph{reduce} the context. Since we consider only finite structures in this work, the following statements hold unconditionally:
We call an object $g\in G$ \emph{reducible} if an object set $X\subseteq G$ with $g\not\subseteq X $ and $g'=X'$ exists.
Otherwise, we call $g$ \emph{irreducible}. 
The definition applys analogously to the attribute set.
A formal context $\K$ containing no reducible attributes and objects is called \emph{reduced}. 
Such a context is called the \emph{standard context} of $\underline{\BB}(\K)$ it is unique up to isomorphism.
The concept lattice of a \emph{standard context} $\K$ is isomorphic to the lattice of every context that is constructed by adding reducible objects or attributes to $\K$.

An element of a lattice $v\in\UL$ is called $\bigvee$-reducible if $v=\bigvee\{x\in\UL\mid x< v\}$.
If $v=\bigwedge\{x\in\UL\mid x> v\}$ holds, $v$ is called $\bigwedge$-reducible.
Otherwise $v$ is called $\bigvee$-irreducible/ $\bigwedge$-irreducible.
The set of all $\bigvee$-irreducible ($\bigwedge$-irreducible) of a lattice $\UL$ is denoted by $J(\UL)$ ($M(\UL)$).
The formal context $(J(\UL),M(\UL)),\le$ is isomorphic to the standard context of the lattice $\UL$.

A part of a formal context $\K=(G,M,I)$ can 
is called
a \emph{subcontext} $\mathbb{S}=(H,N,J)$ with $H\subseteq G$, $N\subseteq M$ and $J= I\cap (H \times N)$.
We use the notions $\mathbb{S}\le\K$ and $[H,N]\coloneqq (H,N,I\cap (H\times N))$ to enable a better readability.
A special kind of subcontext is a \emph{Boolean subcontext of dimension $k$}  which is isomorphic to the context $(\{1,\dots,k\},\{1,\dots,k\},\ne)$.
A subcontext $[H,N]\le\K$ is called \emph{compatible subcontext} if for every concept $(A,B)\in\underline{\BB}(\K)$ the pair $(A\cap H,B\cap N)$ is a concept of $[H,N]$.
An expansion $J\supseteq I$ of the incidence relation is called \emph{block relation} of the formal context $\K=(G,M,I)$, if $g^J$ is an intent in $\K$ and $m^J$ is an extent in $\K$ for all $g\in G$ and respectively all $m\in M$.

For a (concept) lattice $\underline{L}=(L,\leq)$ and $S\subseteq L$, we call $\underline{S}=(\text{S},\le\cap(S\times S))$ \emph{suborder} of $\underline{L}$, and denote this by $\underline{S}\leq \underline{L}$.
Suborders can be generated by single elements:
The \emph{ideal} of $c\in L$ is defined as $(c]:=\{x\in L\mid x\le c\}$;
The \emph{filter} of $c\in L$ is defined as $[c):=\{x\in L\mid c\le x\}$.
For two elements $c,d \in L$ with $c\le d$  we call $[c,d]:=\{x\in L\mid c\le x\le d\}$ the \emph{interval} between $c$ and $d$.
A special kind of suborder are the \emph{crowns} of order $k\ge 3$:
A crown of order $k$ is a partially ordered set $\{x_1,y_1,\dots,x_k,y_k\}$ in which $x_i\le y_i$ for $i\in\{1,\dots,k\}$ and $x_i\le y_{i+1}$ for $i\in\{1,\dots,k-1\}$ and $x_1 \le y_k$ are the only relations in the order.






\section{Related Work}
\label{sec:related-work}

For $\bigwedge$-sublattices, $\bigvee$-sublattices and lattices -- and in general for algebras (i.e., a set with operations defined on its elements) -- homomorphism, congruence relation, and factor algebra are defined explicitly.
In the field of lattice theory, lattice congruences, as defined in~\cite{berghammer2012ordnungen}, where the requirement is compatibility with suprema and infima for finite sets, are examined.
In the realm of ordered sets, no such operations can be utilized.
Thus, there are different approaches to expand these theoretical aspects on ordered sets.

Moorth and Karpagavalli introduce a congruence relation on partially ordered sets that is not a lattice congruence~\cite{ordercongruence}. 
In particular, the congruence classes of this relation do not have to be intervals, a strong change of the original structure is valid.
Other approaches, as given by Snasel and Jukl~\cite{snavsel2009congruences} or Kolibiar~\cite{kolibiar1987congruence} aim to define congruence relations on ordered sets, that are precisely the lattice congruence if applied to lattices.

In the field of Formal Concept Analysis, a plenitude of approaches aim to reduce the size of a data set.
A standard method is the alteration of the object set, the attribute set, or the incidence relation of a formal context or combinations of those approaches.
E.g. in~\cite{Kumar}, a procedure based on random projection is provided;
Dias and Vierira~\cite{diasreducing} explore the replacement of similar objects by a prototypical one, and
a complementary approach relating on the attribute set is investigated by Kuitche et~al.~\cite{Kuitche2018}.
In~\cite{qi2019multi}, the principles of granulation, as introduced in~\cite{zadeh1997toward}, are applied to consider different levels of accuracy in formal contexts and concept lattices.
Besides alteration, there are methods based on identifying and selecting meaningful substructures.
Considering a formal context, Hanika et al.~\cite{hanika2019relevant} select the most relevant attributes related to their ability to reflect the distribution of the objects in the concepts.
Measuring the attributes on their appearance in contranominal scales, D{\"u}rrschnabel et al.~\cite{Durrschnabel} select a set of attributes. 
Another approach is the direct selection of entire concepts.
This can be done by random sampling~\cite{sampling} or the application of various measures.
To this end, Kuznetsov~\cite{kuzuetsov1990stability} introduced a stability measure for formal concepts based on the sets of extent subsets that generate the same intent.
The support measure of association rule mining was borrowed by Stumme et~al.~\cite{STUMME2002189} to generate iceberg concept lattices, special sub-semilattices of the original concept lattice.
Since we want to preserve existing substructures of a lattice via factorization, we turn away from measuring the importance of different parts of a formal context or a (concept) lattice in general. 


%

\section{Imploding Intervals with Congruences and Tolerances}
\label{sec:fact_substr}

This section presents two methods for lattice factorization to implode selected intervals while preserving certain structural properties of the original lattice.
Since we utilize approaches of FCA, we phrase the statements mostly for concept lattices.
However, all statements can be translated to finite lattices in general.

\subsection{Complete Congruence Relations}
\label{sec:congruenz}

Following Ganter and Wille~\cite{fca-book}, we define a (complete) congruence relation of a complete lattice $\UL$ as an equivalence relation $\theta$ on $\UL$ that satisfies the following condition:
$x_t\theta y_t \text{ for } t\in T \Rightarrow (\bigvee_{t\in T} x_t)\theta (\bigvee_{t\in T} y_t) $ and $(\bigwedge_{t\in T} x_t)\theta(\bigwedge_{t\in T} y_t)$.
Thus, congruence relations preserve the $\bigwedge$ and $\bigvee$ operators of $\UL$ in $\UL/ \theta$.

For every lattice $\UL$ and every interval $\US\le\UL$ at least one congruence relation $\theta$ on $\UL$ exists %
that implodes $\US$, meaning that $f\colon\UL\rightarrow \UL/\theta$ with $f(x)=[x]\theta$ (the equivalence class of $\theta$ including x) is an implosion of $\US$ in $\UL$.
This is always the case for the trivial congruence relation $\theta$ that has a single $\theta$-class $[x]\theta=L$,
meaning that $|f(\UL)|=1$.
To utilize this method to our aim of imploding specific intervals while preserving as much of the remaining structure as possible, the following question arises:
Given a lattice $\UL$ and an interval $\US\le\UL$, which congruence relation $\theta$ on $\UL$
is the finest (meaning that $|\UL/\theta|$ is as large as possible) 
so that $f\colon\UL\rightarrow \UL/\theta$ is an implosion of $\US$ in $\UL$
and how can we determine this $\theta$?

The congruence relations on a given lattice $\UL$ are a closure system. 
So a unique finest congruence with the required property exists.
Also, in the finite case, the congruence relations and the compatible subcontexts of the reduced formal context $\K$ with $\B\K\cong \UL$ have a one-to-one correspondence~\cite{fca-book}. 
We adapt this statement to our question setting as follows:

\begin{lemma}
	\label{lem:congruence}
	Let $\mathbb{K}=(G,M,I)$ be a reduced formal context 
	and $\US=[(A,B),(C,D)]\le\underline{\BB}(\K)$ an interval.
	Let $H=\{A\cup\{g\in G\mid g\not\in C\}\}$ and $N=\{D\cup\{m\in M\mid m\not\in B\}\}$.
	The set of all compatible subcontexts $[O,P]\le \K$ with $O\subseteq N$ and $P\subseteq H$ corresponds to the set of all congruence relations $\theta$ on $\underline{\BB}(\K)$ 
	with $f\colon\UL\rightarrow \UL/\theta$ is an implosion of $\US$ in $\UL$.
	The largest compatible subcontext $[O,P]\le \K$ with $O\subseteq N$ and $P\subseteq H$ corresponds to the finest of those congruence relations.
\end{lemma}

\begin{proof}
	The compatible subcontexts of $\K$ correspond to the complete congruences of $\BB(\K)$ so that they induce the complete congruences, and the ordered set of the compatible subcontexts is dually isomorphic to the congruence lattice~\cite[p.107]{fca-book}.
	Due to the order of the congruence lattice, for a given interval $\US$, the set of all congruence relations $\theta$ on $\UL$ with 
	$f\colon\UL\rightarrow \UL/\theta$ is an implosion of $\US$ in $\UL$
	are an order filter,
	meaning $\exists [x]\theta$ with $\US\subseteq [x]\theta$.
	This filter is generated by the unique finest complete congruence with this property.
	Analogously, there is a corresponding compatible subcontext that is the unique greatest one generating an order ideal of all compatible subcontexts corresponding to the former mentioned congruences.
	For every object $g$ of a compatible subcontext the concept $(g'',g')$ is the smallest element of a $\theta$-class of the corresponding congruence $\theta$. 
	Analogously for every attribute $m$ of a compatible subcontext the concept $(m',m'')$ is the greatest element of a $\theta$-class of the corresponding congruence $\theta$~\cite[Prop. 40]{fca-book}.
	Thus, the compatible subcontexts corresponding to the congruences with a $\theta$-class containing $\US$ must not contain the object set $\{\{g\in G \mid (g'', g')\in \US\}\setminus A\}$ and the attribute set $\{\{m\in M \mid (m', m'')\in \US\}\setminus D\}$.
	So the greatest compatible subcontext fulfilling this requirement corresponds to the finest congruence with a $\theta$-class containing $\US$.
	However, the possible selection of objects and attributes can be reduced as follows: 
	Let $g\not\in H$. 
	Then $(g'',g')\le(C,D)$ and $(g'',g')\not\le(A,B)$ and therefore $(g'',g')\wedge(C,D)=(g'',g')$ and $(g'',g')\wedge(A,B)<(g'',g')$ hold.
	So $(g'',g')$ is not the smallest element of a $\theta$-class, and therefore $g$ is not contained in the compatible subcontext.  
\end{proof}

\begin{figure}[t!]
	\centering
	\begin{minipage}{0.47\textwidth}
		\centering
		{\unitlength 0.6mm
			\begin{picture}(60,80)%
			\put(0,0){%
				
				\begin{diagram}{60}{80}
				\Node{1}{40}{0}
				\Node{2}{15}{15}
				\Node{3}{65}{15}
				\Node{4}{40}{30}
				\Node{5}{28}{55}
				\Node{6}{52}{55}
				\Node{7}{40}{80}
				\Node{8}{3}{40}
				\Node{9}{77}{40}
				\Node{10}{15}{65}
				\Node{11}{65}{65}		
				\Node{12}{27}{40}
				\Node{13}{53}{40}
				\Node{14}{56}{68}
				\Node{15}{60}{73}
				
				\Edge{1}{2}
				\Edge{1}{3}
				\Edge{3}{4}
				\Edge{2}{4}
				\Edge{4}{5}
				\Edge{4}{6}
				\Edge{7}{5}
				\Edge{7}{6}
				\Edge{8}{5}
				\Edge{8}{2}
				\Edge{9}{6}
				\Edge{9}{3}
				\Edge{10}{7}
				\Edge{8}{10}
				\Edge{9}{11}
				\Edge{12}{2}
				\Edge{12}{6}
				\Edge{12}{10}
				\Edge{13}{3}
				\Edge{13}{5}
				\Edge{13}{11}
				\Edge{11}{14}
				\Edge{11}{15}
				\Edge{14}{7}
				\Edge{15}{7}

				\leftObjbox{2}{3}{1}{4}
				\rightObjbox{3}{3}{1}{3}
				\leftObjbox{8}{3}{1}{2}
				\rightObjbox{9}{3}{1}{5}
				\leftObjbox{12}{3}{1}{1}
				\rightObjbox{13}{3}{1}{6}
				\NoDots\leftObjbox{14}{3}{3}{7}
				\rightObjbox{15}{3}{0}{8}
				
				\leftAttbox{10}{3}{1}{a}
				\leftAttbox{5}{3}{1}{c}
				\leftAttbox{6}{3}{1}{b}
				\NoDots\leftAttbox{14}{3}{-2}{d}
				\rightAttbox{15}{3}{1}{e}				
				
				\end{diagram}}
			
			\put(15,65){\ColorNode{red}}		
			\put(15,15){\ColorNode{red}}
			\put(28,55){\ColorNode{red}}
			\put(52,55){\ColorNode{red}}
			\put(40,80){\ColorNode{red}}
			\put(27,40){\ColorNode{red}}
			\put(40,30){\ColorNode{red}}
			\put(3,40){\ColorNode{red}}	
			
			\put(65,15){\ColorNode{blue}}
			\put(65,65){\ColorNode{blue}}
			\put(60,73){\ColorNode{blue}}
			\put(77,40){\ColorNode{blue}}
			\put(53,40){\ColorNode{blue}}
			
			\put(37,27){\dashbox{1}(6,6)}
			\put(25,52){\dashbox{1}(6,6)}
			\put(49,52){\dashbox{1}(6,6)}
			\put(37,77){\dashbox{1}(6,6)}
			\put(62,12){\dashbox{1}(6,6)}
			\put(74,37){\dashbox{1}(6,6)}
			\put(50,37){\dashbox{1}(6,6)}
			\put(62,62){\dashbox{1}(6,6)}
			\put(53,65){\dashbox{1}(6,6)}
			\put(57,70){\dashbox{1}(6,6)}
			\end{picture}}			
	\end{minipage}
	\begin{minipage}{0.47\textwidth}
		\centering
		\setlength{\tabcolsep}{4pt}
		\begin{cxt}%
			\att{a}%
			\att{b}%
			\att{c}%
			\att{d}%
			\att{e}%
			\obj{xxBBB}{1} %
			\obj{xBxBB}{2} %
			\obj{.xxxx}{3} %
			\obj{xxxBB}{4} %
			\obj{.xBxx}{5} %
			\obj{.Bxxx}{6} %
			\obj{.BBxB}{7} %
			\obj{.BBBx}{8} %
		\end{cxt}
	\end{minipage}
%
%
%
%
%
%
%
%
%
%
%
%
%
%
%
%
%
	\caption{
	A (concept) lattice $\UBB(\K)$(left) and its corresponding reduced formal context $\K=\GMI$ (right).	
	The intervals $\US_1$(red) and $\US_2$(blue) are highlighted in $\UBB(\K)$.
	The finest congruence relation that implodes $\US_2$
	partitions the concepts in two intervals, the one highlighted with dotted boxes and the remaining one.
	Adding the $\bullet$-marked incidences to $I$ results in the block relation that corresponds to the finest tolerance relation that implodes $\US_2$.
	}
	\label{fig:running_exp}
\end{figure}
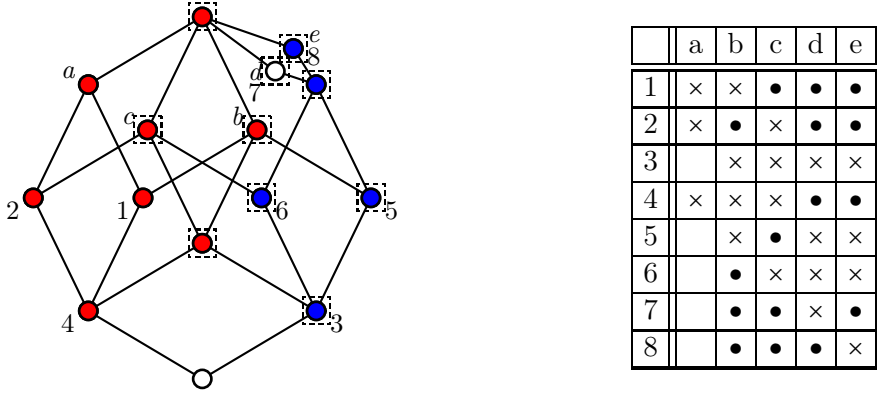

\begin{example}
	Considering the red highlighted interval $\US_1$ 
	in~\cref{fig:running_exp}(left), the sets $H=\{4\}$, $N=\{d,e\}$ result from applying~\cref{lem:congruence}.
	The compatible subcontexts of the corresponding reduced formal context $\K$ (top left) are $[\emptyset,\emptyset]$, $[3,a]$, and $[G,M]$.
	Thus, the compatible subcontext corresponding to the finest (and the only) congruence relation $\theta$, 
	that implodes $\US_1$ is $[\emptyset,\emptyset]$.
	So, the congruence relation we looked for is the trivial one that contains every concept in the same equivalence class.
	For $\US_2$ the sets $H=\{1,2,3,4,7\}$ and $N=\{a,e\}$ arise.
	Thus, the compatible subcontext $[3,a]$ corresponds to the finest congruence relation $\theta$ 
	that implodes $\US_2$.
	$\theta$ partitions the concepts in two intervals, the one highlighted with dotted boxes and the remaining one.
\end{example}

\subsection{Complete Tolerance Relations}
\label{sec:tolleranz}

Another tool for lattice factorization is the (complete) tolerance relation
$\theta$, a reflexive and symmetric relation on a complete lattice $\UL$ that satisfies the following condition:
$x_t\theta y_t \text{ for } t\in T \Rightarrow (\bigvee_{t\in T} x_t)\theta (\bigvee_{t\in T} y_t) \text{ and }(\bigwedge_{t\in T} x_t)\theta(\bigwedge_{t\in T} y_t)$.
As discovered by Czedli~\cite{czedli1982factor} they also generate a factor lattice $\UL/\theta$,
in which, similar to the congruence relations, the $\bigwedge$ and $\bigvee$ operators of $\UL$ are preserved.

Since a tolerance relation does not have to be transitive, we can not expect to find a lattice homomorphism between $\UL$ and $\UL/\theta$.
Instead, we have to arrange with a sublattice homomorphism.
The two possible maps are given by Ganter and Wille in~\cite[Prop. 56]{fca-book}.
This entails to consider sublattice homomorphism as implosions in this chapter, meaning $\exists [x]\theta$ with $\US\subseteq [x]\theta$.

Since every congruence relation is a tolerance relation, the trivial tolerance relation to implode a given interval exists on every lattice.
Thus, the question of finding the finest tolerance relation with this property also arises.
The tolerance relations of a lattice are in an one-to-one correspondence with the block relations of its corresponding formal context.
Hence, we propose the following statement to search for the finest block relation that implodes a chosen interval.

\begin{lemma}
	\label{lem:tollerance}
	Let $\mathbb{K}$ be a reduced formal context and $\US=[(A,B),(C,D)]\le\underline{\BB}(\K)$ an interval.
	Let $\widetilde{J}=I\cup (C\times B)$.
	The set of all block relations $J\supseteq I$ with $J\supseteq\widetilde{J}$ corresponds to the set of all tolerance relations $\theta$ on $\underline{\BB}(\K)$ with 
	with $\UBB(\K)\rightarrow \UBB(\K)/\theta$ being an implosion of $\US$.
	The finest block relation $J$ with $J\supseteq\widetilde{J}$ corresponds to the finest of those congruence relations.
\end{lemma}

\begin{proof}
	Since the set of all block relations is a closure system~\cite[p.122]{fca-book} there is a unique finest block relation $J$ which includes $\widetilde{J}$.
	The tolerance relation $\theta$ corresponding to $J$ has a $\theta$-class containing $\US$ due to the initial inclusion of $(C\times B)$ in $J$~\cite[Thm. 15]{fca-book}.
\end{proof}

In \cref{alg1} we give a strategy to find those relations:
Given a finite lattice $\UL$ and an interval $\US=[(A,B),(C,D)]\le \UL$, in the first step $(C\times D)$ is added to the incidence relation to ensure, that $(A,B)$ and $(C,D)$ are in the same equivalence-class~\cite[Thm. 15]{fca-book} and therefore are mapped to the same element by the factorization.
Then for every object $g\in C$ and every attribute $m\in B$ it is checked whether it satisfies the conditions for block relations with the new incidence relation $J=I\cup (C\times B)$. 
If this is not the case for an object $g$, for the smallest intent $N\subseteq M$ in $\UBB(\K)$ with $g^J\subset N$ all incidences $(g,n)$ with $n\in N\setminus g^J$ are added to $J$.
The method for attributes is analogous.
This process is repeated iteratively until the conditions for block relations hold for every object and attribute.
Note that since the intersection of two intents is an intent itself, the smallest intent selected in every step is unique. The same holds for extents.

\begin{algorithm}[t]
	\SetKw{break}{break}
	\SetKw{or}{or}
	\SetKwFor{loop}{loop}{}{}
	\small \SetKwComment{Comment}{}{} \SetKw{Kwin}{in}
	\DontPrintSemicolon \SetAlgoVlined 
	\KwIn{$\mathbb{K}=(G,M,I)$, $\US=[(A,B),(C,D)]$} 
	\KwOut{$J$}
	
	$J\coloneqq I\cup (C\times B)$\\
	$ext\coloneqq \{H|(H,N)\in\B\K\}$\\
	$int\coloneqq \{N|(H,N)\in\B\K\}$\\
	$check\coloneqq C\cup B$\\
	
	\While{$|check|>0$}{
				
				$x\coloneqq first(check)$
				
				\If{$x\in G$}{
					\If{$x^J\not\in int$}{
						
						$candidates\coloneqq \{y| y\in int, x\subset y\}$
						
						\For{$y\in candidates$}{
						$m_y\coloneqq y\setminus x^J$
						}
					$add\coloneqq min_{|m_y|}\{m_y\}$\\
					$J\coloneqq J\cup \{(x,m)|m\in add \}$\\
					$check\coloneqq check\cup add$
						
					}	
				}
				
				\If{$x\in M$}{
					\If{$x^J\not\in ext$}{
						
						$candidates\coloneqq \{y| y\in ext, x\subset y\}$
						
						\For{$y\in candidates$}{
							$g_y\coloneqq i\setminus x^J$
						}
						$add\coloneqq min_{|g_i|}\{g_i\}$\\
						$J\coloneqq J\cup \{(g,x)|g\in add \}$\\
						$check\coloneqq check\cup add$
					}	
				}
				
				$check\coloneqq check\setminus x$
	
	}
	\KwRet{$J$}
	\caption{Generation of the finest block relation to implode $\US$}
	\label{alg1}
\end{algorithm}

\begin{lemma}
	\Cref{alg1} results in $J$, the finest block relation that implodes $\US$ .
\end{lemma}

\begin{proof}
	$J$ is a block relation, since for every $g\in G$ and every $m\in M$ with $(g,m)\in J$ and $(g,m)\not\in I$ holds that $g^J$ is an intent in $\K$ and $m^J$ is an extent in $\K$.
	Also the tolerance relation $\theta$ corresponding to $J$ has a $\theta$-class containing $\US$ due to the initial inclusion of $(C\times B)$ in $J$~\cite[Thm. 15]{fca-book}.
	
	Since the set of all block relations is a closure system~\cite[p.122]{fca-book} there is a unique finest block relation $\theta$ with the requested properties.
	Assume $\widetilde{J}$ with $\widetilde{J}\subset J$ to be this finest block relation so that $I\subset\widetilde{J}\subset J$.
	Let $(g,m)\in J$ with $(g,m) \not\in \widetilde{J}$ be the first incidence that is added by the algorithm to $J$ while $\widetilde{J}$ does not contain it.
	In each iteration of the strategy, the unique minimal intent containing the current derivation of $g$ is selected as the new derivation of $g$.
	Following $J\subseteq\widetilde{J}$ holds.
	This means $\widetilde{J}=J$.	
\end{proof}

Note that in some cases, the addition of the incidences $(C\times B)$ results in the wanted outcome already. 
In general, this is not the case, e.g., see~\cref{expl:tollerance}.

\longversion{
\begin{example}
	\label{expl:tollernce_1}
	Consider the lattice given in~\cref{fig:congruence}(right) and $\US$, the interval highlighted red.
	The corresponding formal context $\K=(G,M,I)$ (left) is examined to find the finest block relation and, therefore, the finest tolerance relation. 
	$(C,D)=(\{1,2,3,5,6\},d)$ and $(A,B)=(3,\{a,b,c,d,e,g\})$ are the top and the bottom element of $\US$.
	Therefore, $\widetilde{J}=I\cup \{1,2,3,5,6\}\times\{a,b,c,d,e,g\} \cup 3\times d$.
	For the attributes hold $a^{\widetilde{J}}=b^{\widetilde{J}}=c^{\widetilde{J}}=d^{\widetilde{J}}=g^{\widetilde{J}}=d^I$, $e^{\widetilde{J}}=\emptyset^I$ and $f^{\widetilde{J}}=f^I$.
	For the objects $1^{\widetilde{J}}=2^{\widetilde{J}}=3^{\widetilde{J}}=6^{\widetilde{J}}=3^I$, $4^{\widetilde{J}}=4^I$ and $5^{\widetilde{J}}=\emptyset^I$ hold.
	So no further incidences have to be added to generate a block relation and $\widetilde{J}$ (pictured with $\bullet$) corresponds to the finest tolerance relation imploding $\US$.
\end{example}
}

\begin{example}
	\label{expl:tollerance}
	 As for the lattice $\underline{\BB}(\K)$ given in~\cref{fig:running_exp}(left) and the red highlighted interval $\US_{1}$ the corresponding formal context $\K=(G,M,I)$ (right) is examined to find the finest tolerance relation $\theta$ imploding $\US_1$.
	 Since $\US_1=[(\{a\}, \{a,b,c\}),(G,\emptyset)]$ the incidence relation $\widetilde{J}=I\cup G\times \{a,b,c\}$ is generated.
	 After this step, the conditions for block relations have to be checked iteratively.
	 As for the attribute set the condition holds for every attribute since $a^{\widetilde{J}}=b^{\widetilde{J}}=c^{\widetilde{J}}=\emptyset^I$, $d^J=d^I$ and $e^{\widetilde{J}}=e^I$.
	 For the objects $1,2,\dots,6$ also the condition holds.
	 This is not the case for the objects $7$ and $8$.
	 Therefore the incidences $(7,e)$ and $(8,d)$ have to be added to $\widetilde{J}$.
	 After this step, the attributes $e$ and $d$ have to be considered again.
	 Thus the finest block relation $J$ with $\widetilde{J}\subseteq J$ is $J=G\times M$. This block relation corresponds to the trivial tolerance relation.
	 
	 When considering the blue highlighted interval $\US_2=[(\{3\}, \{b,c,d,e\}),(\{3,5,6,8\},\{e\})]$,  we have $\widetilde{J}=I\cup(\{3,5,6,8\}\times\{b,c,d,e\})$.
	 The finest block relation $J$ with $\widetilde{J}\subseteq J$ is depicted in~\cref{fig:running_exp}(right) by the additional $\bullet$-incidences.
\end{example}

\section{Interval Factorization of Lattices}
\label{sec:new_relation}

Our goal in this section is to generate, from a set of intervals $\US_1,\dots,\US_k\le\UL$, a factorization $\UL/\theta$ that can be obtained by an
implosion $f$ of the intervals such that $f$ is injective on $\UL\setminus \bigcup_{i=1}^k\US_i$ (i.e. $|f(\UL\setminus\bigcup_{i=1}^k\US_i)|=|\UL\setminus\bigcup_{i=1}^k\US_i|$ holds) and $f(S_i)\not= f(S_j)$ for $i,j\in \{1,\dots,k\}$ with $i\not=j$.
To this end, we present the \emph{interval relation} $\theta$ on $\underline{L}$, which enables us to generate factor sets that implode exactly the chosen intervals.
By restricting the interval relations to \emph{pure intervals} (see~\cref{sec:lattice_generating_interval_relations}) the generation of lattices is guaranteed.

\subsection{Interval Relations}

To overcome the problem of imploding more than the selected interval, we now introduce a new equivalence relation.

\begin{definition}
	\label{def:new_relation}
	Let $\UL$ be 
	an ordered set 
	and $\{\US_1,\US_2,\dots,\US_k\}$ a set of pairwise disjoint intervals of $\UL$.
	We call the equivalence relation 
	\[\theta_{\US_1,\US_2,\dots,\US_k}\coloneqq \bigcup_{i=1}^k S_i\times S_i \cup \{(x,x)\mid x\in \UL\}\]
	an \emph{interval relation on $\UL$}.
	If $k=1$ holds, we call $\theta=\theta_{\US_1}$ a \emph{1-generated} interval relation.
	For an interval relation $\theta\coloneqq \theta_{\US_1,\US_2,\dots,\US_k}$ we denote the factor set with $\UL/\theta\coloneqq\{[x]\theta\mid x\in \UL\}$ and the equivalence classes of the interval relation by $[x]\theta\coloneqq\{y\in \UL\mid x\theta y\}$. 
\end{definition}

Note that $\theta$ truly is an equivalence relation.
The reflexivity is provided by $\{(x,x)\mid x\in \UL\}$, the symmetry is given by using only the $\times$ operator and the transitivity is based on $\US_1,\US_2,\dots,\US_k$ being disjoint intervals. 

Since each equivalence class is an interval, it includes its supremum and its infimum.
We denote them by $x^\theta\coloneqq\bigvee[x]\theta$ and $x_\theta\coloneqq\bigwedge[x]\theta$, respectively.
We also use the notations $[\US]_{\theta}$ and $[\US]^{\theta}$ for the infimum and supremum of the equivalence class $\{x\in\UL\mid x\in \US\}$ that is generated by $\US$.

Note that every congruence relation on a complete lattice $\UL$ is an interval relation on $\UL$ as well.
More precisely, it is a special case of the \emph{lattice interval relations} that are defined in~\cref{def:lattice_interval_relation}.

\begin{lemma}
	\label{lem:beidedef}
	Let $\theta$ be an equivalence relation on lattice $\UL$.
	The following statements are equivalent:
	\begin{itemize}
		\item[a)] $\theta$ is an interval relation.
		\item[b)] The two following conditions hold for all $x_1,x_2,y_1,y_2\in \UL$:
		\begin{itemize}
			\item[i)] $x_1\theta x_2 \Rightarrow (x_1\vee x_2)\theta x_1$ and $(x_1\wedge x_2)\theta x_1$
			\item [ii)] $x_1\theta x_2$ and $y_1\theta y_2$, $(x_1,y_1)\not\in\theta$ and $ x_1>y_1 \Rightarrow x_2\not<y_2$
		\end{itemize}
	\end{itemize}
\end{lemma}

\begin{proof}
	$a) \Rightarrow b)$ Let $\theta=\theta_{\US_1,\US_2,\dots,\US_k}$ be an interval relation on $\UL$.
	If $x_1\theta x_2$ then $x_1, x_2\in \US_i$ for some $i$ and therefore $(x_1\vee x_2),(x_1\wedge x_2) \in \US_i$ hold. 
	Therefore i) holds.
	Let $x_1, x_2\in \US_i$ and $y_1, y_2\in \US_j$ with $\US_i\not =\US_j$ and $x_1>y_1$.
	Assumed $x_2<y_2$.
	If $x_1\le x_2$ then $y_1<x_1\le x_2<y_2$ and thus $\US_i=\US_j$ holds. \Lightning
	
	In the case of $x_1> x_2$ and $y_1\ge y_2$ we have $x_2<y_2\le y_1<x_1$ also $\US_i=\US_j$. \Lightning
	
	If $x_1> x_2$ and $y_1< y_2$ then $y_2$ and $x_1$ are both upper bounds of $y_1$ and $x_2$. Since $\UL$ is a lattice we have $y_1\vee x_2\le y_2$ and $y_1\vee x_2\le y_2$ and therefore $y_1\vee x_2\in \US_i$ and $y_1\vee x_2\in \US_j$. \Lightning
	
	In all cases, this is a contradiction to the assumptions.
	Hence ii) holds.\\
	$b) \Rightarrow a)$: Let $\theta$ be an equivalence relation on $\UL$ as in b).
	For an arbitrary equivalence class $[x]\theta$ the supremum $x^\theta$ and the infimum $x_\theta$ exist in $[x]\theta$ because of i).
	Let $y\in [x_\theta,x^\theta]$ with $y\not\in [x]\theta$. 
	Then $y\in [x_\theta,x^\theta] \Rightarrow y\le x^\theta,y\ge x_\theta$.
	This is a contradiction to ii).
	So the equivalence classes of $\theta$ are intervals.	
\end{proof}

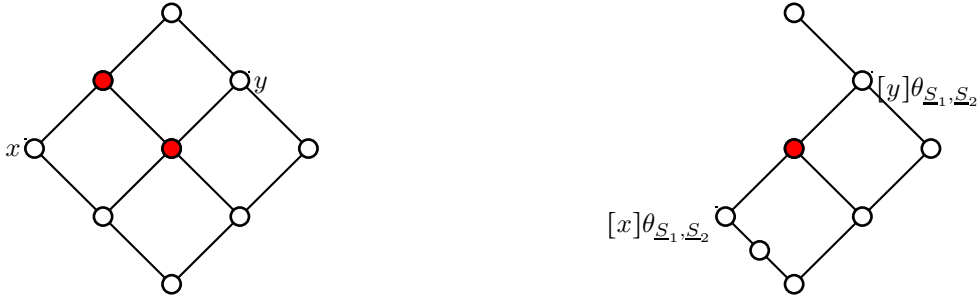
\begin{figure}[t]
	\begin{minipage}{0.49\textwidth}
		\centering
		{\unitlength 0.6mm
			\begin{picture}(60,60)%
			\put(0,0){%
				\begin{diagram}{60}{60}
				
				\Node{1}{0}{30}
				\Node{2}{15}{15}
				\Node{3}{30}{0}
				\Node{4}{15}{45}
				\Node{5}{30}{30}
				\Node{6}{45}{15}
				\Node{7}{30}{60}
				\Node{8}{45}{45}
				\Node{9}{60}{30}
				
				\Edge{1}{2}
				\Edge{2}{3}
				\Edge{1}{4}
				\Edge{2}{5}
				\Edge{3}{6}
				\Edge{4}{5}
				\Edge{5}{6}
				\Edge{4}{7}
				\Edge{5}{8}
				\Edge{6}{9}
				\Edge{7}{8}
				\Edge{8}{9}

				\leftObjbox{1}{3}{-1}{$x$}
				\rightObjbox{8}{3}{-1}{$y$}
				
				\end{diagram}}
			
			\put(15,45){\ColorNode{red}}
			\put(30,30){\ColorNode{red}}			
			
			\end{picture}}			
	\end{minipage}
	\begin{minipage}{0.49\textwidth}
		\centering
		{\unitlength 0.6mm
			\begin{picture}(60,60)%
			\put(0,0){%
				\begin{diagram}{60}{60}
				
				\Node{1}{15}{15}
				\Node{2}{22.5}{7.5}
				\Node{3}{30}{0}
				\Node{4}{30}{30}
				\Node{5}{45}{15}
				\Node{6}{30}{60}
				\Node{7}{45}{45}
				\Node{8}{60}{30}
				
				\Edge{1}{2}
				\Edge{2}{3}
				\Edge{1}{4}
				\Edge{3}{5}
				\Edge{5}{4}
				\Edge{4}{7}
				\Edge{5}{8}
				\Edge{6}{7}
				\Edge{8}{7}				
				
				\leftObjbox{1}{3}{-1}{$[x]\theta_{\US_1,\US_2}$}
				\rightObjbox{7}{3}{-1}{$[y]\theta_{\US_1,\US_2}$}				
				
				\end{diagram}}
			
			\put(30,30){\ColorNode{red}}			
			
			\end{picture}}			
	\end{minipage}
	\caption{Example of a lattice $\UL$ (left) with a red highlighted interval $\US\le\UL$ and two elements $x,y$ that are not comparable. In the factor set $\UL/\theta_{\US}$ (right) the equivalence class corresponding to $\US$ is highlighted in red. Now the elements $[x]\theta_{\US}$ and $[y]\theta_{\US}$ are comparable.
	}
	\label{fig:kleine_treppe}
\end{figure}

\subsubsection{Defining $\le_{\theta}$ for 1-generated interval relations}
\ \\
Since implosions are defined as order preserving maps, on the factor set $\UL/\theta$ the question remains which order the factorization should impose.
To answer this 
we now define the relation $\le_{\theta}$ on the factor set $\UL/\theta$. 
We start with the case of a 1-generated interval relation $\theta=\theta_{\US}$ as presented in~\cref{def:faktorization_order_relation}.
Our construction is motivated by the aim to preserve all comparabilities of $\UL$ in $\UL/\theta$.
So, for an element $x\in \UL$ that is smaller than at least one element of $\US$, it should also $[x]\theta \le_{\theta}[\US]\theta$ hold.
Dually, for an element $y\in \UL$ that is larger than any element of $\US$ should also $[\US]\theta\le_{\theta}[y]\theta$ hold.
Therefore also the elements $[x]\theta$ and $[y]\theta$ should become comparable in $\UL/\theta$.
This is illustrated in~\cref{fig:kleine_treppe}.

To define the order $\le_{\theta}$ with the required properties, we first have to define some areas of 
an ordered set 
generated by an interval.

\begin{definition}
	\label{def:interval}
	Let $\UL$ be 
	an ordered set
	and $\US$ an interval of $\UL$.
	We define the sets 
	\begin{align*}
	S^{\uparrowtail }&\coloneqq \{x\in \UL\setminus\US\mid \exists y\in \US : y<x\},\\
	S^{\downarrowtail}&\coloneqq \{x\in \UL\setminus\US\mid \exists y\in \US : x<y\} 
	\text{ and } \\
	S^{\parallel}&\coloneqq \{x\in \UL\setminus\US \mid \nexists y\in \US : y<x \text{ or } x<y\}.
	\end{align*}
\end{definition}

\begin{proposition}
	Let $\UL$ be 
	an ordered set
	and let $\US$ be an interval of $\UL$.
	Then $S$, $S^{\uparrowtail}$, $S^{\downarrowtail}$ and $S^{\parallel}$ are pairwise disjoint and together cover $\UL$.
	In other words,	$\{S,S^{\uparrowtail},S^{\downarrowtail},S^{\parallel}\}$ is a partition of the elements in $\UL$, with possibly empty classes.
\end{proposition}

\begin{proof}
	Let $\US=[u,v]$.
	Then $S^{\uparrowtail}=[u)\setminus S$ and $S^{\downarrowtail}=(v]\setminus S$.
	Since $[u,v]$ is an interval $S^{\uparrowtail}\cap S^{\downarrowtail}=\emptyset$.
	Also $S^{\parallel}=L\setminus(S\cup S^{\uparrowtail} \cup S^{\downarrowtail})$ holds.
\end{proof}

An example for this division of the 
order 
elements can be seen in~\cref{fig:congruence_not_pure}.

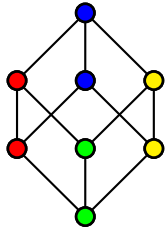
\begin{figure}[t]
	\centering
	{\unitlength 0.6mm
		\begin{picture}(30,45)%
		\put(0,0){%
			\begin{diagram}{30}{45}
			
			\Node{1}{15}{0}
			\Node{2}{0}{15}
			\Node{3}{15}{15}
			\Node{4}{30}{15}
			\Node{5}{0}{30}
			\Node{6}{15}{30}
			\Node{7}{30}{30}
			\Node{8}{15}{45}

			\Edge{1}{2}
			\Edge{1}{3}
			\Edge{1}{4}
			\Edge{2}{5}
			\Edge{2}{6}
			\Edge{3}{5}
			\Edge{3}{7}
			\Edge{4}{6}
			\Edge{4}{7}
			\Edge{5}{8}
			\Edge{6}{8}
			\Edge{7}{8}	
			
			\end{diagram}}
		
		\put(0,15){\ColorNode{red}}
		\put(0,30){\ColorNode{red}}
		
		\put(15,0){\ColorNode{green}}
		\put(15,15){\ColorNode{green}}
		
		\put(15,30){\ColorNode{blue}}
		\put(15,45){\ColorNode{blue}}
		
		\put(30,15){\ColorNode{yellow}}
		\put(30,30){\ColorNode{yellow}}
		
		\end{picture}}			
	
	\caption{For the interval $\US$ (red), the sets $S^{\uparrowtail}$ (blue), $S^{\downarrowtail}$ (green) and $S^{\parallel}$ (yellow) are highlighted.
		In this example, all four sets are intervals. They are not pure but build up a lattice-generating interval relation.
	}
	\label{fig:congruence_not_pure}
\end{figure}

\begin{definition}
	\label{def:faktorization_order_relation}
	Let $\UL$ be 
	an ordered set
	and $\US\le\UL$ an interval.
	We define $[x]\theta\le_{\theta}[y]\theta:\Leftrightarrow (x_\theta\le y^\theta$ 
	or $x\in S^{\downarrowtail}, y\in S^{\uparrowtail})$.
\end{definition}

\begin{example}
	Utilizing congruence or tolerance relations to implode the red highlighted interval in~\cref{fig:running_exp}, the trivial factor lattice, consisting of a single element, arose. With the new construction, we preserve everything outside of the interval:
	In~\cref{fig:running_exp_teil2} the original lattice $\UBB(\K)$ and the factor set $\UBB(\K)/\theta_{\US}$ for the interval relation $\theta_{\US}$ are depicted.
	In $\UBB(\K)/\theta_{\US}$ the interval $\US$ implodes, but the rest of the lattice is preserved. On the other hand, not all joins and meets are preserved.
\end{example}

\begin{figure}[t]
	\centering
	\begin{minipage}{0.33\textwidth}
		\centering
		{\unitlength 0.6mm
			\begin{picture}(60,80)%
			\put(0,0){%
				
				\begin{diagram}{60}{80}
				\Node{1}{40}{0}
				\Node{2}{15}{15}
				\Node{3}{65}{15}
				\Node{4}{40}{30}
				\Node{5}{28}{55}
				\Node{6}{52}{55}
				\Node{7}{40}{80}
				\Node{8}{3}{40}
				\Node{9}{77}{40}
				\Node{10}{15}{65}
				\Node{11}{65}{65}		
				\Node{12}{27}{40}
				\Node{13}{53}{40}
				\Node{14}{56}{68}
				\Node{15}{60}{73}
				
				\Edge{1}{2}
				\Edge{1}{3}
				\Edge{3}{4}
				\Edge{2}{4}
				\Edge{4}{5}
				\Edge{4}{6}
				\Edge{7}{5}
				\Edge{7}{6}
				\Edge{8}{5}
				\Edge{8}{2}
				\Edge{9}{6}
				\Edge{9}{3}
				\Edge{10}{7}
				\Edge{8}{10}
				\Edge{9}{11}
				\Edge{12}{2}
				\Edge{12}{6}
				\Edge{12}{10}
				\Edge{13}{3}
				\Edge{13}{5}
				\Edge{13}{11}
				\Edge{11}{14}
				\Edge{11}{15}
				\Edge{14}{7}
				\Edge{15}{7}

				\NoDots\leftObjbox{4}{3}{-1}{a}
				\NoDots\rightObjbox{3}{3}{1}{c}
				\NoDots\leftObjbox{13}{3}{1}{b}
				\NoDots\leftObjbox{5}{3}{-2}{d}
				
				%
				%
				%
				%
				\end{diagram}}

			\put(15,65){\ColorNode{red}}		
			\put(15,15){\ColorNode{red}}
			\put(28,55){\ColorNode{red}}
			\put(52,55){\ColorNode{red}}
			\put(40,80){\ColorNode{red}}
			\put(27,40){\ColorNode{red}}
			\put(40,30){\ColorNode{red}}
			\put(3,40){\ColorNode{red}}	
			
			\end{picture}}			
	\end{minipage}
	\begin{minipage}{0.45\textwidth}
		\centering
		
		{\unitlength 0.6mm
			
			\begin{picture}(80,80)%
			\put(0,0){%
				
				\begin{diagram}{80}{80}
				\Node{1}{40}{0}
				\Node{2}{65}{15}
				\Node{3}{40}{80}
				\Node{4}{77}{40}
				\Node{5}{65}{65}		
				\Node{6}{53}{40}
				\Node{7}{56}{68}
				\Node{8}{60}{73}

				\Edge{1}{2}
				\Edge{4}{2}
				\Edge{4}{5}
				\Edge{6}{2}
				\Edge{6}{5}
				\Edge{5}{7}
				\Edge{5}{8}
				\Edge{7}{3}
				\Edge{8}{3}

				\NoDots\leftObjbox{3}{3}{-1}{[a]$\theta$=[d]$\theta$}
				\NoDots\rightObjbox{2}{3}{1}{[c]$\theta$}
				\NoDots\leftObjbox{6}{3}{1}{[b]$\theta$}
				
				
				\end{diagram}}
			
			\put(40,80){\ColorNode{red}}	
			
			\end{picture}}			
	\end{minipage}
	%
	%
	%
	%
	%
	%
	%
	%
	%
	\caption{
		A (concept) lattice $\underline{\BB}(\K)$ with a pure interval $\US$ (highlighted red) on the left.
		The lattice $\UBB(\K)/\theta_{\US}$ is pictured on the right.
		Some elements of the lattices are labeled.
	}
	\label{fig:running_exp_teil2}
\end{figure}
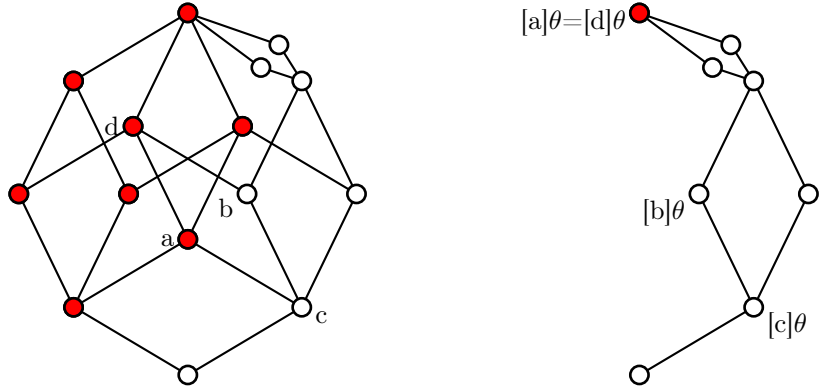

Following, we take a look into the order properties of the factor set with the previously defined relation. For a lattice $\UL$, in general $\UL/\theta$ is neither a lattice (see \cref{fig:kein_verband} below) nor even an ordered set (see \cref{fig:keine_ordnung}).
However, considering only one interval with size larger one, $\UL/\theta$ is always an ordered set: 

\begin{figure}[t]
\begin{minipage}{0.47\textwidth}
	\centering	
	\begin{minipage}{0.49\textwidth}
		\centering
		{\unitlength 0.6mm
			\begin{picture}(30,45)%
			\put(0,0){%
				\begin{diagram}{30}{45}
				
				\Node{1}{15}{0}
				\Node{2}{0}{15}
				\Node{3}{15}{15}
				\Node{4}{30}{15}
				\Node{5}{0}{30}
				\Node{6}{15}{30}
				\Node{7}{30}{30}
				\Node{8}{15}{45}

				\Edge{1}{2}
				\Edge{1}{3}
				\Edge{1}{4}
				\Edge{2}{5}
				\Edge{2}{6}
				\Edge{3}{5}
				\Edge{3}{7}
				\Edge{4}{6}
				\Edge{4}{7}
				\Edge{5}{8}
				\Edge{6}{8}
				\Edge{7}{8}	
				
				\end{diagram}}
			
			\put(0,15){\ColorNode{red}}
			\put(0,30){\ColorNode{red}}
			
			\end{picture}}			
	\end{minipage}
	\begin{minipage}{0.49\textwidth}
		\centering
		{\unitlength 0.6mm	
			\begin{picture}(30,45)%
			\put(0,0){%
				
				\begin{diagram}{30}{45}
				\Node{1}{15}{0}
				\Node{2}{15}{15}
				\Node{3}{30}{15}
				\Node{4}{15}{30}
				\Node{5}{30}{30}
				\Node{6}{15}{45}
				\Node{7}{15}{22.5}
				
				\Edge{1}{2}
				\Edge{1}{3}
				\Edge{2}{5}
				\Edge{3}{5}
				\Edge{4}{3}
				\Edge{4}{6}
				\Edge{5}{6}		
				\Edge{2}{7}
				\Edge{4}{7}		
				
				\end{diagram}}
			
			\put(15,22.5){\ColorNode{red}}
				
			\end{picture}}			
	\end{minipage}
	
	\caption{A lattice $\UL$ with a red highlighted interval $\US\le\UL$ (left). The factor set $\UL/\theta_{\US}$ (right) is no lattice but still an ordered set. The equivalence class $[\US]\theta$ is highlighted in red.
	}
	\label{fig:kein_verband}
\end{minipage}
\hspace{.04\linewidth}
\begin{minipage}{0.47\textwidth}
	\begin{minipage}{0.49\textwidth}
		\centering
		{\unitlength 0.6mm
			\begin{picture}(30,45)%
			\put(0,0){%
				\begin{diagram}{30}{45}
				
				\Node{1}{15}{0}
				\Node{2}{0}{15}
				\Node{3}{15}{15}
				\Node{4}{30}{15}
				\Node{5}{0}{30}
				\Node{6}{15}{30}
				\Node{7}{30}{30}
				\Node{8}{15}{45}

				\Edge{1}{2}
				\Edge{1}{3}
				\Edge{1}{4}
				\Edge{2}{5}
				\Edge{2}{6}
				\Edge{3}{5}
				\Edge{3}{7}
				\Edge{4}{6}
				\Edge{4}{7}
				\Edge{5}{8}
				\Edge{6}{8}
				\Edge{7}{8}	
				
				\end{diagram}}
			
			\put(0,15){\ColorNode{red}}
			\put(0,30){\ColorNode{red}}
			
			\put(30,15){\ColorNode{green}}
			\put(15,30){\ColorNode{green}}
			
			\put(30,30){\ColorNode{blue}}
			\put(15,15){\ColorNode{blue}}
			
			\end{picture}}			
	\end{minipage}
	\begin{minipage}{0.49\textwidth}
		\centering
		{\unitlength 0.6mm	
			\begin{picture}(30,45)%
			\put(0,0){%
				
				\begin{diagram}{30}{45}
				\Node{1}{15}{0}
				\Node{2}{0}{22.5}
				\Node{3}{15}{22.5}
				\Node{4}{30}{22.5}				
				\Node{5}{15}{45}

				\Edge{1}{2}
				\Edge{1}{3}
				\Edge{1}{4}
				\Edge{2}{3}
				\Edge{3}{4}
				\Edge{2}{4}
				\Edge{2}{5}		
				\Edge{3}{5}
				\Edge{4}{5}		
				
				\end{diagram}}
			
			\thicklines
			\put(15,0){\vector(-15,22.5){13}}
			\put(15,0){\vector(0,1){20}}
			\put(15,0){\vector(15,22.5){13}}
			
			\put(0,22.5){\vector(15,22.5){13}}
			\put(15,22.5){\vector(0,1){20}}
			\put(30,22.5){\vector(-15,22.5){13}}
			
			\put(0,22.5){\vector(1,0){14}}
			\put(15,22.5){\vector(-1,0){14}}
			\put(15,22.5){\vector(1,0){14}}
			\put(30,22.5){\vector(-1,0){14}}
			
			\qbezier(0,22.5)(15,30)(30,22.5)
			\put(4,24.2){\vector(-2,-1){1}}
			\put(26,24.2){\vector(2,-1){1}}
			
			\put(0,22.5){\ColorNode{red}}
			\put(15,22.5){\ColorNode{green}}
			\put(30,22.5){\ColorNode{blue}}
			\put(15,0){\ColorNode{white}}
			
			\end{picture}}			
	\end{minipage}
	
	\caption{``Penrose crown'': a lattice $\UL$  with three pairwise comparable intervals $\US_1$(red), $\US_2$(green), $\US_3$(blue)$\le\UL$ (left). The factor set $\UL/\theta_{\US_1,\US_2,\US_3}$ (right) is not even an ordered set. 
	It holds $[\US_{1}]\theta\le[\US_{2}]\theta\le[\US_{3}]\theta\le[\US_{1}]\theta$.
	Thus $\le_{\theta}$ is not anti-symmetric and, therefore no order relation.
	}
	\label{fig:keine_ordnung}
\end{minipage}
\end{figure}

\begin{lemma}
	\label{lem:ordnung}
	Let $\theta=\theta_{\US}$ be an interval relation on the ordered set $\UL$.
	Then $\le_{\theta}$	is an order on $\UL/\theta$.
\end{lemma}

\begin{proof}
	Reflexivity: 
	$\le$ is an order on $\UL$, meaning $x\le x$ for all $x\in \UL$ and especially $x_{\theta}\le x\le x^{\theta}$. Therefore $[x]\theta\le_{\theta} [x]\theta$ holds in $\UL/\theta$.
	
	Transitivity: 
	Let $[x]\theta\le_{\theta}[y]\theta$ and $[y]\theta\le_{\theta}[z]\theta$ in $\UL/\theta$. 
	In the case that $x_{\theta}\le y^{\theta}$ and $y_{\theta}\le z^{\theta}$ in $\UL$ either $y_{\theta}=y^{\theta}$ and therefore $x_{\theta}\le z^{\theta}$. 
	Otherwise, in the case of $y_{\theta}\not=y^{\theta}$ we have $y\in \US$, $x\in S^{\downarrowtail}\cup S$ and $z \in S^{\uparrowtail}\cup S$ and then $[x]\theta\le_{\theta}[z]\theta$.
	In the case $x\in S^{\downarrowtail}, y\in S^{\uparrowtail}$ and $y_\theta\le z^\theta$
	holds $z\in S^{\uparrowtail}$ and therefore $[x]\theta\le_{\theta}[z]\theta$.
	The case $y\in S^{\downarrowtail}, z\in S^{\uparrowtail}$ and $x_\theta\le y^\theta$ is analogous.
	
	Anti-symmetry: 
	Let $[x]\theta\le_{\theta}[y]\theta$ and $[y]\theta\le_{\theta}[x]\theta$.
	If $x_{\theta}\le y^{\theta}$ and $y_{\theta}\le x^{\theta}$, follows $[x]\theta=[y]\theta$ directly since either $x_{\theta}=x^{\theta}$ or $x_{\theta}=x^{\theta}$ or $x_{\theta}=y_{\theta}$ and $x^{\theta}=y^{\theta}$.
	The cases $x\in S^{\downarrowtail}, y\in S^{\uparrowtail}$, $y_{\theta}\le x^{\theta}$ and  $y\in S^{\downarrowtail}, x\in S^{\uparrowtail}$, $x_{\theta}\le y^{\theta}$ can not occur as well as the case of  $x\in S^{\downarrowtail}, y\in S^{\uparrowtail}$,  $y\in S^{\downarrowtail}, x\in S^{\uparrowtail}$.
\end{proof}


\begin{lemma}
	Let $\UL$ be 
	an ordered set
	and $\theta=\theta_{\US}$ an interval relation on $\UL$.
	Then $\le_{\theta}$	is the smallest order on $\UL/\theta$ so that the map $\varphi\colon\UL\rightarrow\UL/\theta,~x\mapsto[x]\theta$ is surjective and order-preserving.
\end{lemma}

\begin{proof}
	The surjectivity follows directly from the fact that for every $[x]\theta\in \UL/\theta$ we have $x\in \UL$ as representative.
	The order is preserved due to the definition of $\le_{\theta}$:
	If $x\le y$ in $\UL$ holds, we have $x_{\theta}\le y^{\theta}$ and therefore $[x]\theta\le_{\theta} [y]\theta$ in $\UL/\theta$.
	
	We show that $\le_{\theta}$ is the smallest of those relations by contraposition:
	Let $R$ be an order on $\UL/\theta$ so that $\varphi$ is surjective and order-preserving.
	Assumed there are $x,y\in \UL$ with 
	$[x]\theta\le_{\theta} [y]\theta$ and $[x]\theta \not R [y]\theta$.
	Since $\varphi$ is order-preserving we know that $x\not\le y$ in $\UL$ and moreover no element of the $\theta$-class $[x]\theta$ is less or equal an element of the $\theta$-class $[y]\theta$.
	In particular, $x_{\theta}\not\le y^{\theta}$ holds in $\UL$.
	Thus, $x\in S^{\downarrowtail}$ and $y\in S^{\uparrowtail}$ due to the definition of $\le_{\theta}$.
	Therefore, we have $x_{\theta}\le [\US]^{\theta}$ and $[\US]_{\theta}\le y^{\theta}$ and due to $\varphi$ being order-preserving also 
	$[x]\theta R [\US]\theta$ and $[\US]\theta R [y]\theta$.
	Consequently, $[x]\theta R [y]\theta$ follows from the transitivity of the order $R$. \Lightning	
\end{proof}

\subsubsection{Defining $\le_{\theta}$ in general}
\ \\
Considering more than one interval with size larger one, we define the relation $\le_{\theta}$ on $\UL/\theta$, by generalizing~\cref{def:faktorization_order_relation} 
as follows:

\begin{definition}
	\label{def:faktorization_preorder_relation}
	Let $\UL$ be 
	an ordered set
	and $\theta\coloneqq\theta_{\US_1, \dots, \US_k}$ an interval relation on $\UL$.
	We define the relation
	$[x]\theta\le_{\theta}[y]\theta\colon\Leftrightarrow (x_\theta\le y^\theta$ 
	or $\exists i_1,\dots,i_l \in \{1,\dots,k\}$ with	
	$x_{\theta}\in S_{i_1}^{\downarrowtail}$, 
	$[\US_{i_1}]_{\theta}\in S_{i_2}^{\downarrowtail}$, 
	\dots , 
	$[\US_{i_{l-1}}]_{\theta}\in S_{i_l}^{\downarrowtail}$,
	$y^{\theta}\in S_{i_l}^{\uparrowtail})$.
\end{definition}

The relation is illustrated in~\cref{fig:treppe}.

\begin{figure}[t]
	\begin{minipage}{0.49\textwidth}
		\centering
		{\unitlength 0.6mm
			\begin{picture}(90,90)%
			\put(0,0){%
				\begin{diagram}{90}{90}
				
				\Node{1}{0}{30}
				\Node{2}{15}{15}
				\Node{3}{30}{0}
				\Node{4}{15}{45}
				\Node{5}{30}{30}
				\Node{6}{45}{15}
				\Node{7}{30}{60}
				\Node{8}{45}{45}
				\Node{9}{60}{30}
				\Node{10}{45}{75}
				\Node{11}{60}{60}
				\Node{12}{75}{45}
				\Node{13}{60}{90}
				\Node{14}{75}{75}
				\Node{15}{90}{60}
				
				\Edge{1}{2}
				\Edge{2}{3}
				\Edge{1}{4}
				\Edge{2}{5}
				\Edge{3}{6}
				\Edge{4}{5}
				\Edge{5}{6}
				\Edge{4}{7}
				\Edge{5}{8}
				\Edge{6}{9}
				\Edge{7}{8}
				\Edge{8}{9}
				\Edge{7}{10}	
				\Edge{8}{11}	
				\Edge{9}{12}	
				\Edge{10}{11}	
				\Edge{11}{12}	
				\Edge{10}{13}	
				\Edge{11}{14}	
				\Edge{12}{15}	
				\Edge{13}{14}	
				\Edge{14}{15}

				\leftObjbox{1}{3}{-1}{$x$}
				\rightObjbox{8}{3}{-1}{$y$}
				\rightObjbox{15}{3}{-1}{$z$}

				\end{diagram}}
			
			\put(15,45){\ColorNode{red}}
			\put(30,30){\ColorNode{red}}
			
			\put(75,45){\ColorNode{green}}
			\put(60,60){\ColorNode{green}}

			\end{picture}}			
	\end{minipage}
	\begin{minipage}{0.49\textwidth}
		\centering
		{\unitlength 0.6mm
			\begin{picture}(90,90)%
			\put(0,0){%
				\begin{diagram}{90}{90}
				
				\Node{1}{15}{15}
				\Node{2}{22.5}{7.5}
				\Node{3}{30}{0}
				\Node{4}{30}{30}
				\Node{5}{45}{15}
				\Node{6}{30}{60}
				\Node{7}{45}{45}
				\Node{8}{60}{30}
				\Node{9}{45}{75}
				\Node{10}{60}{60}
				\Node{11}{60}{90}
				\Node{12}{67.5}{82.5}
				\Node{13}{75}{75}
				
				\Edge{1}{2}
				\Edge{2}{3}
				\Edge{1}{4}
				\Edge{3}{5}
				\Edge{5}{4}
				\Edge{4}{7}
				\Edge{5}{8}
				\Edge{6}{7}
				\Edge{8}{7}
				\Edge{6}{9}	
				\Edge{7}{10}	
				\Edge{9}{10}	
				\Edge{9}{11}		
				\Edge{10}{13}	
				\Edge{11}{12}	
				\Edge{12}{13}

				\leftObjbox{1}{3}{-1}{$[x]\theta_{\US_1,\US_2}$}
				\rightObjbox{7}{3}{-1}{$[y]\theta_{\US_1,\US_2}$}
				\rightObjbox{13}{3}{-1}{$[z]\theta_{\US_1,\US_2}$}

				\end{diagram}}
			
			\put(30,30){\ColorNode{red}}
			\put(60,60){\ColorNode{green}}

			\end{picture}}			
	\end{minipage}
	\caption{In the lattice $\UL$ (left) -- with 
		intervals $\US_1,\US_2\le\UL$ 
		highlighted in red and green, respectively --
		the with three elements $x,y,z$ that are not comparable. In the factor set $\UL/\theta_{\US_1,\US_2}$ (right) the equivalence classes corresponding to $\US_1$ and $\US_2$ are highlighted red and green, respectively. 
		Now the elements $[x]\theta_{\US_1,\US_2},[y]\theta_{\US_1,\US_2}$ and $[z]\theta_{\US_1,\US_2}$ are comparable.
	}
	\label{fig:treppe}
\end{figure}
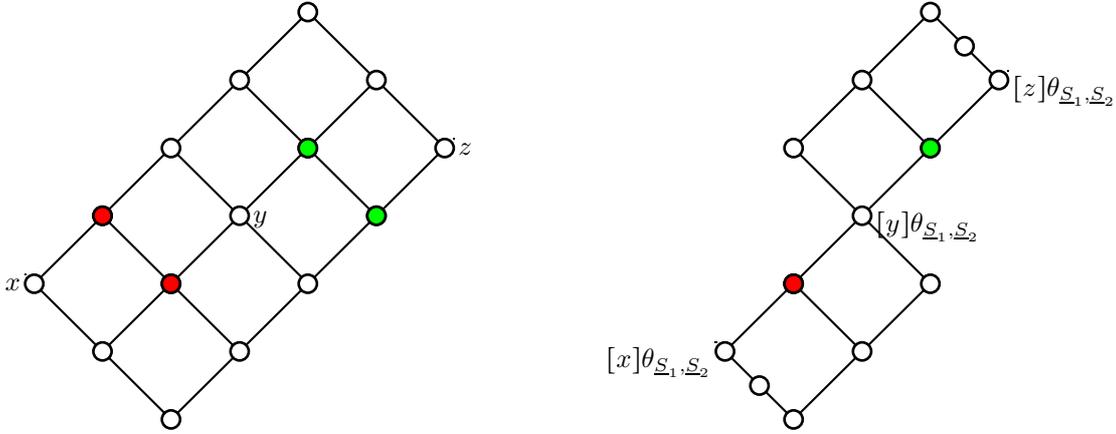

Note that $\le_{\theta}$ for one interval in~\cref{def:faktorization_order_relation} is a special case of $\le_{\theta}$ in~\cref{def:faktorization_preorder_relation}.
Therefore we use the notion $\le_{\theta}$ in the following both for imploding one as well as multiple intervals.

We can show that considering several intervals, $\le_{\theta}$ is always a preorder on $\UL/\theta$, i.e., it is always reflexive and symmetric but ot necessarily anti-symmetric:

\begin{lemma}
	\label{lem:preorder}
	Let $\UL$ be 
	an ordered set 
	and $\theta=\theta_{\US_1, \dots, \US_k}$ an interval relation on $\UL$.
	Then $\le_{\theta}$ is a preorder on $\UL/\theta$.
\end{lemma}

\begin{proof}
	Reflexivity: 
	$\le$ is an order on $\UL$, meaning $x\le x$ for all $x\in \UL$ and especially $x_{\theta}\le x\le x^{\theta}$. Therefore $[x]\theta\le_{\theta} [x]\theta$ holds in $\UL/\theta$.
	
	Transitivity: 
	Let $[x]\theta\le_{\theta}[y]\theta$ and $[y]\theta\le_{\theta}[z]\theta$ in $\UL/\theta$.
	If $[x]\theta=[y]\theta$ or $[y]\theta=[z]\theta$ the statement is similar to the proof of \cref{lem:ordnung}.
	Assume $[x]\theta\not=[y]\theta\not=[z]\theta$.
	If $x_{\theta}\le y^{\theta}$ and $y_{\theta}\le z^{\theta}$ in $\UL$ 
	then either $y^{\theta}=y_{\theta}$ and therefore $x_{\theta}\le z^{\theta}$
	or there is an interval $\US_i$ with $y\in\US_i$.
	In this case $x_{\theta} \in \US_{i}^{\downarrowtail}$ and $z^{\theta}\in \US_{i}^{\uparrowtail}$ hold. 
	In both cases follows $[x]\theta\le_{\theta}[z]\theta$.
	If $x_{\theta}\le y^{\theta}$ and $\exists \US_{i_1}, \dots, \US_{i_l}$ as described with $y_{\theta}\in \US_{i_1}^{\downarrowtail}$ and $z^{\theta}\in \US_{i_l}^{\uparrowtail}$
	then $x_{\theta} \in \US^{\downarrowtail}$ and
	$[\US]_{\theta}\in \US_{i_1}^{\downarrowtail}$ for the interval $\US=([y]\theta,\le)$.
	Then $[x]\theta\le_{\theta}[z]\theta$.
	The case of $y_{\theta}\le z^{\theta}$ and $\exists \US_{i_1}, \dots, \US_{i_l}$ as described with $x_{\theta}\in \US_{i_1}^{\downarrowtail}$ and $y^{\theta}\in \US_{i_l}^{\uparrowtail})$ follows analogously.
	If $\exists \US_{i_1}, \dots, \US_{i_l}$ and  $\US_{j_1}, \dots, \US_{j_l}$ as described with $x_{\theta}\in \US_{i_1}^{\downarrowtail}$, $y^{\theta}\in \US_{i_l}^{\uparrowtail})$, $y_{\theta}\in \US_{j_1}^{\downarrowtail}$ and $z^{\theta}\in \US_{j_l}^{\uparrowtail})$
	there is an interval $\US=([y]\theta,\le)$ so that
	$[\US_{i_l}]_{\theta}\in \US^{\downarrowtail}$ and
	$[\US]_{\theta}\in \US_{j_1}^{\downarrowtail}$.
	Then $[x]\theta\le_{\theta}[z]\theta$.
	
\end{proof}

Note that $\le_{\theta}$ as given in~\cref{def:faktorization_preorder_relation} is the same relation as 
$[x]\theta\le_{\theta}[y]\theta:\Leftrightarrow (x_\theta\le y^\theta$ 
or $\exists i_1,\dots,i_l \in \{1,\dots,k\}$ with 
$x_{\theta}\in \US_{i_1}^{\downarrowtail}$, 
$[\US_{i_2}]_{\theta}\in \US_{i_1}^{\uparrowtail}$, 
$\dots$ , 
$[\US_{i_l}]_{\theta}\in \US_{i_l-1}^{\uparrowtail}$,
$y^{\theta}\in \US_{i_l}^{\uparrowtail})$.

We observe that $\le_{\theta}$ is the transitive closure of a simpler relation as follows:

\begin{lemma}
	Let $\UL$ be 
	an ordered set and $\theta$ an interval relation on $\UL$.
	Let $[x]\theta\le^*[y]\theta:\Leftrightarrow x_\theta\le y^\theta$ be a relation on $\UL/\theta$.
	Then $\le_{\theta}$ is the transitive closure of $\le^*$.
\end{lemma}

\begin{proof}
	Let $x,y \in \UL$ and $\US_i$ an interval of $\theta$.
	Then $x_{\theta}\in S_{i}^{\downarrowtail}$ is equivalent to $x_{\theta}\le[S_i]^{\theta}$.
	Further, $y^{\theta}\in S_{i}^{\uparrowtail}$ is equivalent to $y^{\theta}\le[S_i]_{\theta}$.
	Since $[S_i]_{\theta}\le [S_i]^{\theta}$, the statement follows directly.
\end{proof}

\subsection{Order-preserving Interval Relations}
Up to now, we have only shown (in~\cref{lem:preorder}) that $\le_{\theta}$ is a preorder and~\cref{fig:keine_ordnung} showed that it will not always be anti-symmetric.
We now investigate the order properties of $\le_{\theta}$ in more detail:

\begin{definition}
	\label{def:ordered_new_relation}
	Let $\UL$ be 
	an ordered set
	and $\theta$ an interval relation on $\UL$.
	We call $\theta$ \emph{order-preserving} on $\UL$ if $(\UL/\theta,\le_{\theta})$ is an ordered set.
\end{definition}

Considering a 1-generated interval relation, from~\cref{lem:preorder} follows directly:

\begin{corollary}
	Let $\UL$ be 
	an ordered set
	and $\theta$ an interval relation on $\UL$.
	If $\theta$ is 1-generated, $\theta$ is order-preserving.
\end{corollary}

For more than one interval, we can provide a necessary and sufficient condition for an interval relation $\theta$ to be order-preserving in~\cref{lem:order_bedingung}.

\begin{definition}
	Let $\UL$ be an ordered set and $2\le k$.
	A set $\{\US_1, \dots, \US_k\}$  of intervals in $\UL$ are called \emph{Penrose crown of order $k$} in $\UL$ if they are pairwise disjoint 
	and if
	$[\US_{1}]_{\theta}\in \US_{2}^{\downarrowtail}$,
	$[\US_{2}]_{\theta}\in \US_{3}^{\downarrowtail}$,  
	$\dots$, 
	$[\US_{k-1}]_{\theta}\in \US_{k}^{\downarrowtail}$,
	$[\US_{k}]_{\theta}\in \US_{1}^{\downarrowtail}$.
\end{definition}

We call such a constellation of intervals a Penrose crown, named after the ``impossible staircase'' created by L.\ Penrose and R.\ Penrose  in 1958 (and previously by O.\ Reutersvärd in 1937)~\cite{penrose1958impossible}.
The construction became popular by M.C.\ Escher's lithograph ``Ascending ans Descending''.
The intervals in the lattice in~\cref{fig:keine_ordnung} form a Penrose crown of order 3. Another example is illustrated in~\cref{fig:crown_general}.

\begin{figure}[t]
	\centering
	{\unitlength 0.6mm
		\begin{picture}(100,20)%
		\put(0,0){%
			\begin{diagram}{100}{20}
			
			\Node{1}{0}{0}
			\Node{2}{20}{0}
			\Node{3}{40}{0}
			\Node{4}{0}{20}
			\Node{5}{20}{20}
			\Node{6}{40}{20}
			\Node{7}{60}{0}
			\Node{8}{60}{20}
			\Node{9}{80}{0}
			\Node{10}{80}{20}
			\Node{11}{100}{0}
			\Node{12}{100}{20}

			\Edge{1}{5}
			\Edge{2}{6}
			\Edge{3}{8}
			\Edge{7}{10}
			\Edge{1}{4}
			\Edge{2}{5}
			\Edge{3}{6}
			\Edge{7}{8}
			\Edge{9}{10}
			\Edge{9}{12}
			\Edge{11}{12}
			\Edge{11}{4}
			
			\leftObjbox{1}{0}{1}{$[S_1]_{\theta}$}
			\leftObjbox{2}{0}{1}{$[S_2]_{\theta}$}
			\leftObjbox{3}{0}{1}{$[S_3]_{\theta}$}
			\leftObjbox{7}{0}{1}{$[S_4]_{\theta}$}
			\leftObjbox{9}{0}{1}{$[S_5]_{\theta}$}
			\leftObjbox{11}{0}{1}{$[S_6]_{\theta}$}
			\NoDots\leftObjbox{4}{0}{-6}{$[S_1]^{\theta}$}
			\NoDots\leftObjbox{5}{0}{-6}{$[S_2]^{\theta}$}
			\NoDots\leftObjbox{6}{0}{-6}{$[S_3]^{\theta}$}
			\NoDots\leftObjbox{8}{0}{-6}{$[S_4]^{\theta}$}
			\NoDots\leftObjbox{10}{0}{-6}{$[S_5]^{\theta}$}
			\NoDots\leftObjbox{12}{0}{-6}{$[S_6]^{\theta}$}
			
			\end{diagram}}
		
		\put(0,0){\ColorNode{red}}
		\put(0,20){\ColorNode{red}}
		
		\put(20,0){\ColorNode{green}}
		\put(20,20){\ColorNode{green}}
		
		\put(40,0){\ColorNode{blue}}
		\put(40,20){\ColorNode{blue}}
		
		\put(60,0){\ColorNode{yellow}}
		\put(60,20){\ColorNode{yellow}}
		
		\put(80,0){\ColorNode{orange}}
		\put(80,20){\ColorNode{orange}}
		
		\put(100,0){\ColorNode{purple}}
		\put(100,20){\ColorNode{purple}}
		
		\end{picture}}	
	
	\caption{Penrose crown of order 6.}		
	\label{fig:crown_general}
\end{figure}
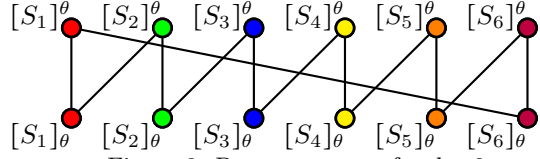

\begin{theorem}
	\label{lem:order_bedingung}
	Let $\UL$ be a finite lattice and $\theta=\theta_{\US_1, \dots, \US_k}$ an interval relation on $\UL$.
	$\theta$ is order-preserving if and only if there exists no Penrose crown $\{\US_{i_1},\dots,\US_{i_l}\}$ in $\UL$
	with
	$i_1,\dots,i_l\in \{1,\dots,k\}$.
\end{theorem}

\begin{proof}
	"$\Rightarrow$":
	Assumed some intervals $\US_{1},\dots,\US_{l}\le \UL$ exist as described.
	Then $[\US_1]\theta\le_{\theta} [\US_l]\theta$ and $[\US_l]\theta\le_{\theta} [\US_1]\theta$ by definition of $\le_{\theta}$.
	Since $[\US_{1}]\not=[\US_{l}]$ the preorder $\le_{\theta}$ is not anti-symmetric and therefore not an order.

	"$\Leftarrow$":
	Assume the relation $\le_{\theta}$ is not an order.
	Then there are two equivalence classes $[S_i]\theta$ and $[S_j]\theta$ with
	$[S_i]\theta\le_{\theta} [S_j]\theta$, $[S_j]\theta\le_{\theta} [S_i]\theta$ and $[S_i]\theta\not=[S_j]\theta$ in $\UL/\theta$.
	For two intervals $\US,\underline{T}$ with $[S]\theta\not=[T]\theta$ it holds that $[S]_{\theta}\le [T]^{\theta}\Rightarrow [S]_{\theta} \in \underline{T}^{\downarrowtail}$.
	Due to the definition of $\le_{\theta}$, one of the following cases has to occur.
	In the case of $[\US_i]_{\theta}\le [\US_j]^{\theta}$ 
	and $[\US_j]_{\theta}\le [\US_i]^{\theta}$ 
	$\{\US_{i},\US_j\}$ is a Penrose crown of order 2.
	%
	In the case of $[S_i]_{\theta}\le [S_j]^{\theta}$ and
	$[S_j]_{\theta}\not\le [S_i]^{\theta}$ we have
	$[S_j]_{\theta}\in \US_{1}^{\downarrowtail}$,
	$[\US_{1}]_{\theta}\in \US_{2}^{\downarrowtail}$,
	$\dots$, 
	$[\US_{l}]_{\theta}\in [S_i]\theta^{\downarrowtail}$
	because of $[S_j]\theta\le_{\theta} [S_i]\theta$.
	The case of 
	$[S_j]_{\theta}\le [S_i]^{\theta}$ and
	$[S_i]_{\theta}\not\le [S_j]^{\theta}$
	follow analogously.
	In the case of 
	$[S_j]_{\theta}\not\le [S_i]^{\theta}$ and
	$[S_i]_{\theta}\not\le [S_j]^{\theta}$ we have
	$[S_j]_{\theta}\in \US_{1}^{\downarrowtail}$, 
	$[\US_{1}]_{\theta}\in \US_{2}^{\downarrowtail}$,
	$\dots$,
	$[\US_{l}]_{\theta}\in [S_i]\theta^{\downarrowtail}$
	and
	$[S_i]_{\theta}\in \US_{m}^{\downarrowtail}$, 
	$[\US_{m}]_{\theta}\in \US_{m+1}^{\downarrowtail}$,
	$\dots$,
	$[\US_{s}]_{\theta}\in [S_j]\theta^{\downarrowtail}$.
\end{proof}

Since a Penrose crown consists of at least two intervals, an interval relation is always order preserving if at most one of its intervals consists of more than one element:

\begin{lemma}
	\label{lem:anzahl_intervalle}
	Let $\UL$ be an ordered set and $\theta=\theta_{\US_1, \dots, \US_k}$ an interval relation on $\UL$.
	If $|[\US_i]\theta|\ge 2$ for at most one $i\in \{1,\dots,k\}$, $\theta$ is an order-preserving interval relation.
\end{lemma}

\begin{proof}
	If $\theta$ includes no interval of size 2 or larger, $\UL/\theta=\UL$.
	If  $\theta=\theta_{\US_1}$ with $|\US_1|\ge 2$ the statement follows from~\cref{lem:ordnung}.
\end{proof}

In the case of $\UL$ being a lattice, such a constellation can not occur if at most two intervals of the interval relation $\theta$ include more than a single element of $\UL$:

\begin{lemma}
Let $\UL$ be a finite lattice and $\theta=\theta_{\US_1, \dots, \US_k}$ an interval relation on $\UL$.
If $|[\US_i]\theta|\ge 2$ for at most two $i\in \{1,\dots,k\}$, $\theta$ is an order-preserving interval relation.
\end{lemma}

\begin{proof}
	If $\theta$ includes no or one interval of size 2 or larger, the proof is the same as in~\cref{lem:anzahl_intervalle}.
	Let $\theta=\theta_{\US_1,\US_2}$ with $|\US_{1}|,|\US_{2}|\ge 2$.
	Assume that $[\US_1]_\theta \in  S_2^{\downarrowtail}$ and $[\US_2]_\theta \in  S_1^{\downarrowtail}$.
	Then $[\US_1]_\theta \vee [\US_2]_\theta \in \US_{1}$ and $[\US_1]_\theta \vee [\US_2]_\theta \in \US_{2}$.
	Hence, the intervals $\US_{1}$ $\US_{2}$ are not disjoint.
	This is a contradiction to $\theta$ being an interval relation.
\end{proof}

Moreover, the case mentioned in~\cref{lem:order_bedingung} can only occur if $\UL$ contains a Penrose crown of order $l\ge 2$.
If $\UL$ is a lattice, it has to contain a Penrose crown of order $l\ge 3$ and therefore a crown of the same order as suborder. 

%
%

\begin{corollary}
	Let $\UL$ be an ordered set and $\theta$ an interval relation on $\UL$.
	If $\UL$ does not contain a Penrose crown of order $l\ge 2$, then $\theta$ is order-preserving.
\end{corollary}

\begin{corollary}
	\label{lem:crown}
	Let $\UL$ be a lattice and $\theta$ an interval relation on $\UL$.
	If $\UL$ does not contain a crown of order $l\ge 3$ as a suborder, then $\theta$ is order-preserving.
\end{corollary}

Since a dismantlable lattice $\UL$ -- meaning the iterative elimination of all doubly irreducible elements results in the elimination of the whole lattice -- never contains a crown~\cite{kelly1974crowns}, every interval relation on such a lattice is order preserving.

\begin{corollary}
	Let $\UL$ be an ordered set and $\theta$ an interval relation on $\UL$.
	If $\UL$ is planar than
	$\theta$ is order-preserving.
\end{corollary}

As shown, for an order preserving interval relation $\theta_{\US_1,\dots, \US_k}$, we have an implosion of the intervals $\US_{1},\dots,\US_{k}$ as defined in~\cref{def:implosion}.
In the following section we investigate the preservation of the lattice properties.



\subsection{Lattice-generating Interval Relations}
\label{sec:lattice_generating_interval_relations}

So far, we have been interested in interval relations with factor sets that are ordered sets.
Now we focus on interval relations where the resulting factor set is even a lattice.
Therefore we 
restrict us to the case where 
$\UL$ is a (finite) lattice.

\begin{definition}
	\label{def:lattice_interval_relation}
	Let $\UL$ be a finite lattice and $\theta$ an interval relation on $\UL$.
	We call $\theta$ a \emph{lattice-generating interval relation} on $\UL$ if $(\UL/\theta,\le_{\theta})$ is a lattice.
\end{definition}


Note that every congruence relation is a lattice-generating interval relation since the equivalence classes form pairwise disjoint intervals and the order  (denoted by $\le_c$ in the following lemma) which is defined on the factor lattice $\UL/\theta$ for a congruence relations $\theta$ is equal to $\le_{\theta}$:

\begin{lemma}
	Let $\UL$ be a complete lattice and $\theta$ a complete congruence relation on $\UL$.
	Then the orders $\le_{\theta}$ and $[x]\theta\le_c [y]\theta\colon\Leftrightarrow x\theta (x\wedge y)$ are identical on $\UL/\theta$.
\end{lemma}

\begin{proof}
	"$\Leftarrow$":
	Let $[x]\theta\le_c [y]\theta$ and therefore $(x\wedge y)\in [x]\theta$.
	Since $x\wedge y\le y$ holds, we have $x_{\theta}\le (x\wedge y) \le y\le y^{\theta}$ and consequently $[x]\theta\le_{\theta}[y]\theta$.
	
	"$\Rightarrow$": Let $[x]\theta\le_{\theta}[y]\theta$.
	In the case of $x_{\theta}\le y^{\theta}$ we have $x\theta x_{\theta}$ and $y\theta y^{\theta}$ and therefore (due to the definition of congruence relations) $(x\wedge y)\theta (x_{\theta}\wedge y^{\theta})=x_{\theta}$.
	Then $[x]\theta\le_c [y]\theta$ holds.
	
	If $x_{\theta}\not\le y^{\theta}$ then there are intervals $\US_{1},\dots,\US_k$ in $\theta$ with 
	$x_{\theta}\le [ \US_1]^{\theta}$, 
	$[ \US_1]_{\theta}\le [ \US_2]^{\theta}$, 
	\dots , 
	$[\US_{k}]_{\theta}\le y^{\theta}$.
	Then we have $[x]\theta\le_c [\US_1]\theta\le_c \dots \le_c [y]\theta$.
\end{proof}

We will now provide a characterization of lattice-generating interval relations. 

\begin{definition}
		\label{def:new_relation2}
		Let $\UL$ be a finite lattice and $\US\le\UL$ an interval. 
		We call $\US$ a \emph{nested interval} of $\UL$ if 
		there are two intervals $\underline{T},\underline{U}\le \UL$ so that $\US,\underline{T},\underline{U}$ are a Penrose crown of order 3 in $\UL$. 
%
		%
		We call $\US$ a \emph{pure interval} of $\UL$ if it is not nested.
\end{definition}


\begin{corollary}
	Let $\UL$ be a finite lattice and $\US\le\UL$ an interval.	
	$\US$ is nested if and only if there are
 	$ x,y\in S^{\parallel}, a\in S^{\uparrowtail}$ and $v\in S^{\downarrowtail}$ with 
	$y=x\vee v,
	x=y\wedge a,
	y\not\le a$ and $v\not\le x$.
\end{corollary}

An example of a lattice with a nested and a pure interval is given in~\cref{fig:nested_interval_darstellung}.
Also, both highlighted intervals in the lattice in~\cref{fig:running_exp} (and therefore the one in~\cref{fig:running_exp_teil2}) are pure. 
In the 1-generated case, the pure intervals are exactly the lattice-generating intervals:

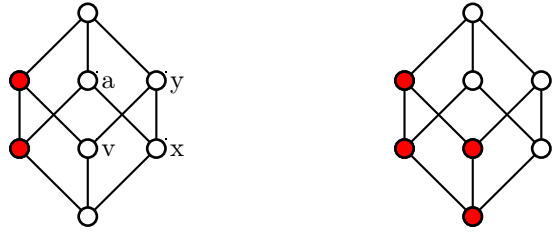
\begin{figure}[t]
	\centering
	\begin{minipage}{0.3\textwidth}
	\centering
	{\unitlength 0.6mm
		\begin{picture}(60,45)%
		\put(0,0){%
			\begin{diagram}{60}{45}
			
			\Node{1}{30}{0}
			\Node{2}{15}{15}
			\Node{3}{30}{15}
			\Node{4}{45}{15}
			\Node{5}{15}{30}
			\Node{6}{30}{30}
			\Node{7}{45}{30}
			\Node{8}{30}{45}

			\Edge{1}{2}
			\Edge{1}{3}
			\Edge{1}{4}
			\Edge{2}{5}
			\Edge{2}{6}
			\Edge{3}{5}
			\Edge{3}{7}
			\Edge{4}{6}
			\Edge{4}{7}
			\Edge{5}{8}
			\Edge{6}{8}
			\Edge{7}{8}	
			
			\rightObjbox{3}{3}{-1}{v}
			\rightObjbox{4}{3}{-1}{x}
			\rightObjbox{7}{3}{-1}{y}
			\rightObjbox{6}{3}{-1}{a}
			
			\end{diagram}}
		
		\put(15,15){\ColorNode{red}}
		\put(15,30){\ColorNode{red}}
		
		\end{picture}}
	\end{minipage}
	\begin{minipage}{0.3\textwidth}
	\centering
	{\unitlength 0.6mm
		\begin{picture}(60,45)%
		\put(0,0){%
			\begin{diagram}{60}{45}
			
			\Node{1}{30}{0}
			\Node{2}{15}{15}
			\Node{3}{30}{15}
			\Node{4}{45}{15}
			\Node{5}{15}{30}
			\Node{6}{30}{30}
			\Node{7}{45}{30}
			\Node{8}{30}{45}

			\Edge{1}{2}
			\Edge{1}{3}
			\Edge{1}{4}
			\Edge{2}{5}
			\Edge{2}{6}
			\Edge{3}{5}
			\Edge{3}{7}
			\Edge{4}{6}
			\Edge{4}{7}
			\Edge{5}{8}
			\Edge{6}{8}
			\Edge{7}{8}	
						
			\end{diagram}}
		
		\put(15,15){\ColorNode{red}}
		\put(15,30){\ColorNode{red}}
		\put(30,15){\ColorNode{red}}
		\put(30,0){\ColorNode{red}}
		
		\end{picture}}
		\end{minipage}
		\caption{
			On the left, a nested interval $\US$  of the lattice $\UL$ is highlighted in red. $\US$ is part of a Penrose crows of order 3 with the Intervals $[x,a]$ and $[v,y]$.
			As presented in~\cref{fig:kein_verband} $\UL/\theta_{\US}$ is no lattice.
			On the right, a pure interval is highlighted in red.
	}
	\label{fig:nested_interval_darstellung}			
\end{figure}

\begin{lemma}
	\label{lem:verband}
	Let $\theta=\theta_{\US}$ be an interval relation on lattice $\underline{L}$.
	Then $\theta$ is lattice-generating if and only if $\US$ is pure.	
\end{lemma}

\begin{proof}
	"$\Rightarrow$":
	We show the contraposition:
	Let $\US$ be a nested interval, meaning 
	$\exists x,y\in S^{\parallel}, a\in S^{\uparrowtail}$ and $v\in S^{\downarrowtail}$ with $y=x\vee v,x=y\wedge a,y\not\le a$ and $v\not\le x$.
	For all elements $d\in S^{\downarrowtail}, e\in S^{\uparrowtail}$ holds $[d]\theta<_{\theta}[e]\theta$ in $\UL/\theta_{\US}$.
	It follows that $[v]\theta<_{\theta}[a]\theta$ in $\UL/\theta_{\US}$ and $\exists [c]\theta\in S^{\uparrowtail}$ with $[v]\theta<_{\theta}[c]\theta\le_{\theta} [a]\theta$ and $[x]\theta<_{\theta} [c]\theta$.
	So $[y]\theta$ and $[c]\theta$ are two different minimal upper bounds of $[v]\theta$ and $[x]\theta$.
	Thus $\UL/\theta_{\US}$ is not a lattice.
	
	"$\Leftarrow$":
	We show the contraposition:
	$\UL/\theta$ is an ordered set by~\cref{lem:ordnung}.
	Suppose $\UL\theta$ is not a lattice.
	Then exist $[x]\theta,[v]\theta$ in $\UL/\theta$ with two smallest upper bounds or two greatest lower bounds.
	Due to the duality of lattices, we only examine the case of two incomparable smallest upper bounds $[a]\theta,[y]\theta$.
	We have $[x]\theta\not =[v]\theta$ in $\UL/\theta$ and thus $x\not= v$ in $\UL$.
	Since $\UL$ is a lattice and the factorization does not affect the order of the elements in $S^{\downarrowtail}, S^{\uparrowtail}$ and $S^{\parallel}$ we have that $v,x$ are not both in the same of those sets.
	Otherwise $[x]\theta\vee [v]\theta= [x\vee v]\theta$ holds.

	In addition, we show that $x\not\in S$:
	If $x\in S$ and $v\in S^{\uparrowtail}$, we have $[x]\theta\vee [v]\theta=[v]\theta$.
	If $x\in S$ and $v\in S^{\downarrowtail}$, we have $[x]\theta\vee [v]\theta=[x]\theta$.
	If $x\in S$ and $v\in S$, we have $[x]\theta\vee [v]\theta=[v]\theta=[x]\theta$.
	If $x\in S$ and $v\in S^{\parallel}$, we have $[x]\theta\vee [v]\theta=[x_{\theta}\vee v]\theta$ or otherwise $\UL$ would not be a lattice.
	Analogous, one can show that $v\not\in S$.
	
	In case of $x$ or $v$ in $S^{\uparrowtail}$ we have:
	W.l.o.g.\ let $v\in S^{\uparrowtail}$.
	If $x\in S^{\downarrowtail}$ we have $[x]\theta\vee [v]\theta=[v]\theta$.
	If $x\in S^{\parallel}$ we have $[x]\theta\vee [v]\theta=[x\vee v]\theta$
	Therefore the only possibility for $[x]\theta$ and $[v]\theta$ having two minimal upper bounds is $x\in S^{\parallel}$ and $v\in S^{\downarrowtail}$ with $v\not\le x$ (or the other way around).
	
	W.l.o.g.\ let $y=v\vee x$ in $\UL$.
	Since $x\in S^{\parallel}$ and $v\in S^{\downarrowtail}$ we have $y\in S^{\uparrowtail}\cup S^{\parallel}$.
	If $y\in S^{\uparrowtail}$ we have $[x]\theta\vee [v]\theta=[x_{\theta}\vee v]\theta\le_{\theta}[y]\theta$ as the supremum of $[x]\theta$ and $[v]\theta$.
	So let $y\in S^{\parallel}$. 
	Then $[y]\theta$ is a smallest upper bound of $[x]\theta$ and $[v]\theta$.
	Let $[a]\theta\not= [y]\theta$ be another smallest upper bound of $[x]\theta$ and $[v]\theta$.
	Then either $x\not\le a$ or $v\not\le a$ in $\UL$.
	Due to the definition of the order in $\UL/\theta$ we have $v\not\le a$, $x\le a$ and $a\in S^{\uparrowtail}$ in $\UL$.
	Then $S$ is a nested interval.
\end{proof}

The example shown in~\cref{fig:kein_verband} illustrates the implosion of a nested interval $\US$ in a lattice $\UL$.
In this case, $\UL/\theta_{\US}$ is an ordered set but no lattice.

Let $\UL$ be a finite lattice and $\US_1,\US_2$ two disjoint pure intervals of $\UL$.
Note, that in general $\US_2$ is not a pure interval in $\UL/\theta_{\US_1}$.
Consequently, an interval relation $\theta=\theta_{\US_1, \dots, \US_k}$ it not necessarily lattice-generating just because it is pure -
\cref{fig:pure_intervals_gegenbeispiel} shows a counterexample.
Also, not every lattice-generating interval relation consists of only pure intervals as can be seen in \cref{fig:congruence_not_pure}.

\begin{figure}[t]
	\centering
	\begin{minipage}{0.3\textwidth}
		\centering
		{\unitlength 0.6mm
			\begin{picture}(45,45)%
			\put(0,0){%
				\begin{diagram}{45}{45}
				
				\Node{1}{22.5}{0}
				\Node{2}{0}{15}
				\Node{3}{15}{15}
				\Node{4}{30}{15}
				\Node{5}{45}{15}
				\Node{6}{0}{30}
				\Node{7}{15}{30}
				\Node{8}{30}{30}
				\Node{9}{45}{30}
				\Node{10}{22.5}{45}

				\Edge{1}{2}
				\Edge{1}{3}
				\Edge{1}{4}
				\Edge{1}{5}
				\Edge{2}{6}
				\Edge{2}{7}
				\Edge{3}{7}
				\Edge{3}{8}
				\Edge{4}{8}
				\Edge{4}{9}
				\Edge{5}{9}
				\Edge{5}{6}
				\Edge{6}{10}
				\Edge{7}{10}
				\Edge{8}{10}
				\Edge{9}{10}
	
				\end{diagram}}
			
			\put(30,15){\ColorNode{red}}
			\put(30,30){\ColorNode{red}}
			\put(15,30){\ColorNode{blue}}
			\put(15,15){\ColorNode{blue}}
	
			\end{picture}}			
	\end{minipage}
	\begin{minipage}{0.3\textwidth}
		\centering
		{\unitlength 0.6mm
			\begin{picture}(45,45)%
			\put(0,0){%
				\begin{diagram}{45}{45}
				
				\Node{1}{22.5}{0}
				\Node{2}{0}{15}
				\Node{3}{15}{15}
				\Node{4}{35}{25}
				\Node{5}{45}{15}
				\Node{6}{0}{30}
				\Node{7}{15}{30}
				\Node{8}{22.5}{45}
				\Node{9}{45}{30}
				
				\Edge{1}{2}
				\Edge{1}{3}
				\Edge{1}{5}
				\Edge{2}{6}
				\Edge{2}{7}
				\Edge{3}{7}
				\Edge{3}{4}
				\Edge{4}{9}
				\Edge{5}{9}
				\Edge{5}{6}
				\Edge{6}{8}
				\Edge{7}{8}
				\Edge{9}{8}
				
				\leftObjbox{2}{3}{-1}{v}
				\leftObjbox{6}{3}{-1}{y}
				\rightObjbox{5}{3}{-1}{x}
				\rightObjbox{9}{3}{-1}{a}	
				
				\end{diagram}}
			
			\put(35,25){\ColorNode{red}}
			\put(15,30){\ColorNode{blue}}
			\put(15,15){\ColorNode{blue}}
			
			\end{picture}}			
	\end{minipage}
	\begin{minipage}{0.3\textwidth}
		\centering
		{\unitlength 0.6mm
			\begin{picture}(45,45)%
			\put(0,0){%
				\begin{diagram}{45}{45}
				
				\Node{1}{22.5}{0}
				\Node{2}{0}{15}
				\Node{3}{10}{20}
				\Node{4}{30}{15}
				\Node{5}{45}{15}
				\Node{6}{0}{30}
				\Node{7}{22.5}{45}
				\Node{8}{30}{30}
				\Node{9}{45}{30}
				
				\Edge{1}{2}
				\Edge{1}{4}
				\Edge{1}{5}
				\Edge{2}{6}
				\Edge{2}{3}
				\Edge{3}{8}
				\Edge{4}{8}
				\Edge{4}{9}
				\Edge{5}{9}
				\Edge{5}{6}
				\Edge{6}{7}
				\Edge{8}{7}
				\Edge{9}{7}
				
				\leftObjbox{2}{3}{-1}{v}
				\rightObjbox{5}{3}{-1}{x}
				\leftObjbox{6}{3}{-1}{y}
				\rightObjbox{9}{3}{-1}{a}		
				
				\end{diagram}}
			
			\put(30,15){\ColorNode{red}}
			\put(30,30){\ColorNode{red}}
			\put(10,20){\ColorNode{blue}}
			
			\end{picture}}			
	\end{minipage}
	\caption{A lattice $\UL$ with two pure intervals that are red and blue highlighted (left).
		If the red interval is factorized, the blue interval becomes nested (middle). 
		If the blue interval is factorized, the red interval becomes nested (right).
	}
	\label{fig:pure_intervals_gegenbeispiel}	
\end{figure}
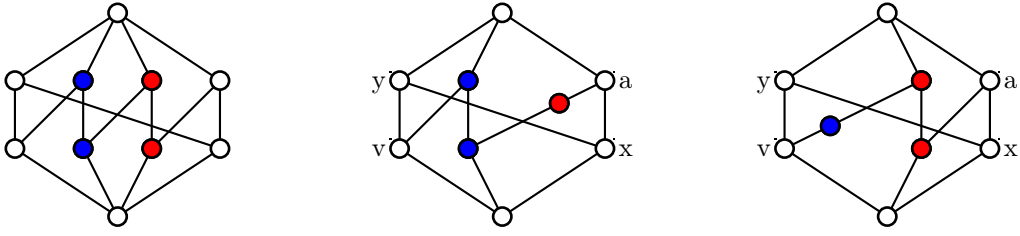

This means that, having an interval relation $\theta_{\US}$ 
with a nested interval $\US$, it is possible to alter $\theta$ in a way that \emph{purifies} it by adding additional intervals.
In this case, it is necessary to make the elements $a$ and $y$ or the elements $x$ and $v$ comparable in the lattice.
Some possibilities to purify a nested interval are illustrated in \cref{fig:purification}.
Note that the purification of an interval can also necessarily require several new intervals since there are possibly more than just one set of elements $a,v,x,y$ that make $\US$ nested.
Also, each additional interval may interact with the other added intervals as well as with $\US$ so that new problematic elements can arise.

\begin{figure}[t]
	\centering
	\begin{minipage}{0.22\textwidth}
		\centering
		{\unitlength 0.6mm
			\begin{picture}(60,45)%
			\put(0,0){%
				\begin{diagram}{60}{45}
				
				\Node{1}{30}{0}
				\Node{2}{15}{15}
				\Node{3}{30}{15}
				\Node{4}{45}{15}
				\Node{5}{15}{30}
				\Node{6}{30}{30}
				\Node{7}{45}{30}
				\Node{8}{30}{45}

				\Edge{1}{2}
				\Edge{1}{3}
				\Edge{1}{4}
				\Edge{2}{5}
				\Edge{2}{6}
				\Edge{3}{5}
				\Edge{3}{7}
				\Edge{4}{6}
				\Edge{4}{7}
				\Edge{5}{8}
				\Edge{6}{8}
				\Edge{7}{8}	
				
				\rightObjbox{3}{3}{-1}{v}
				\rightObjbox{4}{3}{-1}{x}
				\rightObjbox{7}{3}{-1}{y}
				\rightObjbox{6}{3}{-1}{a}
				
				\end{diagram}}
			
			\put(15,15){\ColorNode{red}}
			\put(15,30){\ColorNode{red}}
			\put(45,15){\ColorNode{blue}}
			\put(45,30){\ColorNode{blue}}
			
			\end{picture}}
	\end{minipage}
	\begin{minipage}{0.22\textwidth}
		\centering
		{\unitlength 0.6mm
			\begin{picture}(60,45)%
			\put(0,0){%
				\begin{diagram}{60}{45}
				
				\Node{1}{30}{0}
				\Node{2}{15}{22.5}
				\Node{3}{30}{15}
				\Node{4}{45}{22.5}
				\Node{5}{30}{45}
				\Node{6}{30}{30}

				\Edge{1}{3}
				\Edge{3}{2}
				\Edge{3}{4}
				\Edge{2}{6}
				\Edge{4}{6}
				\Edge{6}{5}
				
				
				\end{diagram}}
			
			\put(15,22.5){\ColorNode{red}}
			\put(45,22.5){\ColorNode{blue}}
			
			\end{picture}}
	\end{minipage}
\quad\quad\quad\quad
\begin{minipage}{0.22\textwidth}
			\centering
			{\unitlength 0.6mm
				\begin{picture}(60,45)%
				\put(0,0){%
					\begin{diagram}{60}{45}
					
					\Node{1}{30}{0}
					\Node{2}{15}{15}
					\Node{3}{30}{15}
					\Node{4}{45}{15}
					\Node{5}{15}{30}
					\Node{6}{30}{30}
					\Node{7}{45}{30}
					\Node{8}{30}{45}

					\Edge{1}{2}
					\Edge{1}{3}
					\Edge{1}{4}
					\Edge{2}{5}
					\Edge{2}{6}
					\Edge{3}{5}
					\Edge{3}{7}
					\Edge{4}{6}
					\Edge{4}{7}
					\Edge{5}{8}
					\Edge{6}{8}
					\Edge{7}{8}	
					
					\rightObjbox{3}{3}{-1}{v}
					\rightObjbox{4}{3}{-1}{x}
					\rightObjbox{7}{3}{-1}{y}
					\rightObjbox{6}{3}{-1}{a}
					
					\end{diagram}}
				
				\put(15,15){\ColorNode{red}}
				\put(15,30){\ColorNode{red}}
				\put(30,30){\ColorNode{blue}}
				\put(30,45){\ColorNode{blue}}
				
				\end{picture}}
		\end{minipage}
		\begin{minipage}{0.22\textwidth}
			\centering
			{\unitlength 0.6mm
				\begin{picture}(60,45)%
				\put(0,0){%
					\begin{diagram}{60}{45}
					
					\Node{1}{30}{0}
					\Node{2}{45}{30}
					\Node{3}{30}{15}
					\Node{4}{45}{15}
					\Node{5}{15}{30}
					\Node{6}{30}{45}

					\Edge{1}{3}
					\Edge{1}{4}
					\Edge{3}{5}
					\Edge{3}{2}
					\Edge{4}{2}
					\Edge{5}{6}
					\Edge{2}{6}	
					
					
					\end{diagram}}
				
				\put(15,30){\ColorNode{red}}
				\put(30,45){\ColorNode{blue}}
				
				\end{picture}}
	\end{minipage}	
\ \\
\ \\
\ \\
\ \\
\begin{minipage}{0.22\textwidth}
			\centering
			{\unitlength 0.6mm
				\begin{picture}(60,45)%
				\put(0,0){%
					\begin{diagram}{60}{45}
					
					\Node{1}{30}{0}
					\Node{2}{15}{15}
					\Node{3}{30}{15}
					\Node{4}{45}{15}
					\Node{5}{15}{30}
					\Node{6}{30}{30}
					\Node{7}{45}{30}
					\Node{8}{30}{45}

					\Edge{1}{2}
					\Edge{1}{3}
					\Edge{1}{4}
					\Edge{2}{5}
					\Edge{2}{6}
					\Edge{3}{5}
					\Edge{3}{7}
					\Edge{4}{6}
					\Edge{4}{7}
					\Edge{5}{8}
					\Edge{6}{8}
					\Edge{7}{8}	
					
					\rightObjbox{3}{3}{-1}{v}
					\rightObjbox{4}{3}{-1}{x}
					\rightObjbox{7}{3}{-1}{y}
					\rightObjbox{6}{3}{-1}{a}
					
					\end{diagram}}
				
				\put(15,15){\ColorNode{red}}
				\put(15,30){\ColorNode{red}}
				\put(30,15){\ColorNode{blue}}
				\put(30,0){\ColorNode{blue}}
				
				\end{picture}}
		\end{minipage}
		\begin{minipage}{0.22\textwidth}
			\centering
			{\unitlength 0.6mm
				\begin{picture}(60,45)%
				\put(0,0){%
					\begin{diagram}{60}{45}
					
					\Node{1}{30}{0}
					\Node{2}{15}{15}
					\Node{3}{45}{30}
					\Node{4}{45}{15}
					\Node{5}{30}{45}
					\Node{6}{30}{30}
					
					\Edge{1}{2}
					\Edge{1}{4}
					\Edge{2}{6}
					\Edge{4}{6}
					\Edge{4}{3}
					\Edge{6}{5}
					\Edge{3}{5}	
					
					
					\end{diagram}}
				
				\put(15,15){\ColorNode{red}}
				\put(30,0){\ColorNode{blue}}
				
				\end{picture}}
	\end{minipage}
\quad\quad\quad\quad
\begin{minipage}{0.22\textwidth}
			\centering
			{\unitlength 0.6mm
				\begin{picture}(60,45)%
				\put(0,0){%
					\begin{diagram}{60}{45}
					
					\Node{1}{30}{0}
					\Node{2}{15}{15}
					\Node{3}{30}{15}
					\Node{4}{45}{15}
					\Node{5}{15}{30}
					\Node{6}{30}{30}
					\Node{7}{45}{30}
					\Node{8}{30}{45}

					\Edge{1}{2}
					\Edge{1}{3}
					\Edge{1}{4}
					\Edge{2}{5}
					\Edge{2}{6}
					\Edge{3}{5}
					\Edge{3}{7}
					\Edge{4}{6}
					\Edge{4}{7}
					\Edge{5}{8}
					\Edge{6}{8}
					\Edge{7}{8}	
					
					\rightObjbox{3}{3}{-1}{v}
					\rightObjbox{4}{3}{-1}{x}
					\rightObjbox{7}{3}{-1}{y}
					\rightObjbox{6}{3}{-1}{a}
					
					\end{diagram}}
				
				\put(15,15){\ColorNode{red}}
				\put(15,30){\ColorNode{red}}
				\put(30,45){\ColorNode{blue}}
				\put(45,30){\ColorNode{blue}}
				
				\end{picture}}
		\end{minipage}
		\begin{minipage}{0.22\textwidth}
			\centering
			{\unitlength 0.6mm
				\begin{picture}(60,45)%
				\put(0,0){%
					\begin{diagram}{60}{45}
					
					\Node{1}{30}{0}
					\Node{2}{15}{22.5}
					\Node{3}{30}{15}
					\Node{4}{45}{15}
					\Node{5}{30}{45}
					\Node{6}{30}{30}		
					
					\Edge{3}{2}
					\Edge{1}{3}
					\Edge{1}{4}
					\Edge{2}{6}
					\Edge{4}{6}
					\Edge{6}{5}	
					
					
					\end{diagram}}
				
				\put(15,22.5){\ColorNode{red}}
				\put(30,45){\ColorNode{blue}}
				
				\end{picture}}
	\end{minipage}	
\ \\
\ \\
\ \\
\ \\
\begin{minipage}{0.22\textwidth}
	\centering
	{\unitlength 0.6mm
		\begin{picture}(60,45)%
		\put(0,0){%
			\begin{diagram}{60}{45}
			
			\Node{1}{30}{0}
			\Node{2}{15}{15}
			\Node{3}{30}{15}
			\Node{4}{45}{15}
			\Node{5}{15}{30}
			\Node{6}{30}{30}
			\Node{7}{45}{30}
			\Node{8}{30}{45}

			\Edge{1}{2}
			\Edge{1}{3}
			\Edge{1}{4}
			\Edge{2}{5}
			\Edge{2}{6}
			\Edge{3}{5}
			\Edge{3}{7}
			\Edge{4}{6}
			\Edge{4}{7}
			\Edge{5}{8}
			\Edge{6}{8}
			\Edge{7}{8}	
			
			\rightObjbox{3}{3}{-1}{v}
			\rightObjbox{4}{3}{-1}{x}
			\rightObjbox{7}{3}{-1}{y}
			\rightObjbox{6}{3}{-1}{a}
			
			\end{diagram}}
		
		\put(15,15){\ColorNode{red}}
		\put(15,30){\ColorNode{red}}
		\put(45,15){\ColorNode{blue}}
		\put(30,30){\ColorNode{blue}}
		
		\end{picture}}
\end{minipage}
\begin{minipage}{0.22\textwidth}
	\centering
	{\unitlength 0.6mm
		\begin{picture}(60,45)%
		\put(0,0){%
			\begin{diagram}{60}{45}
			
			\Node{1}{30}{0}
			\Node{2}{15}{18}
			\Node{3}{30}{15}
			\Node{4}{45}{23}
			\Node{5}{45}{30}
			\Node{6}{30}{45}

			\Edge{3}{2}
			\Edge{1}{3}
			\Edge{2}{4}
			\Edge{4}{5}
			\Edge{5}{6}	
			
			
			\end{diagram}}
		
		\put(15,18){\ColorNode{red}}
		\put(45,23){\ColorNode{blue}}
		
		\end{picture}}
\end{minipage}
\quad\quad\quad\quad
\begin{minipage}{0.22\textwidth}
		\centering
		{\unitlength 0.6mm
			\begin{picture}(60,45)%
			\put(0,0){%
				\begin{diagram}{60}{45}
				
				\Node{1}{30}{0}
				\Node{2}{15}{15}
				\Node{3}{30}{15}
				\Node{4}{45}{15}
				\Node{5}{15}{30}
				\Node{6}{30}{30}
				\Node{7}{45}{30}
				\Node{8}{30}{45}

				\Edge{1}{2}
				\Edge{1}{3}
				\Edge{1}{4}
				\Edge{2}{5}
				\Edge{2}{6}
				\Edge{3}{5}
				\Edge{3}{7}
				\Edge{4}{6}
				\Edge{4}{7}
				\Edge{5}{8}
				\Edge{6}{8}
				\Edge{7}{8}	
				
				\rightObjbox{3}{3}{-1}{v}
				\rightObjbox{4}{3}{-1}{x}
				\rightObjbox{7}{3}{-1}{y}
				\rightObjbox{6}{3}{-1}{a}
				
				\end{diagram}}
			
			\put(15,15){\ColorNode{red}}
			\put(15,30){\ColorNode{red}}
			\put(30,15){\ColorNode{blue}}
			\put(45,30){\ColorNode{blue}}
			
			\end{picture}}
	\end{minipage}
	\begin{minipage}{0.22\textwidth}
		\centering
		{\unitlength 0.6mm
			\begin{picture}(60,45)%
			\put(0,0){%
				\begin{diagram}{60}{45}
				
				\Node{1}{30}{0}
				\Node{2}{15}{27}
				\Node{3}{45}{22}
				\Node{4}{45}{15}
				\Node{5}{30}{45}
				\Node{6}{30}{30}
				
				\Edge{1}{4}
				\Edge{4}{3}
				\Edge{3}{2}
				\Edge{2}{6}
				\Edge{6}{5}
				
				
				\end{diagram}}
			
			\put(15,27){\ColorNode{red}}
			\put(45,22){\ColorNode{blue}}
			
			\end{picture}}
\end{minipage}
\ \\
\ \\
\ \\
\ \\
\begin{minipage}{0.22\textwidth}
		\centering
		{\unitlength 0.6mm
			\begin{picture}(60,45)%
			\put(0,0){%
				\begin{diagram}{60}{45}
				
				\Node{1}{30}{0}
				\Node{2}{15}{15}
				\Node{3}{30}{15}
				\Node{4}{45}{15}
				\Node{5}{15}{30}
				\Node{6}{30}{30}
				\Node{7}{45}{30}
				\Node{8}{30}{45}

				\Edge{1}{2}
				\Edge{1}{3}
				\Edge{1}{4}
				\Edge{2}{5}
				\Edge{2}{6}
				\Edge{3}{5}
				\Edge{3}{7}
				\Edge{4}{6}
				\Edge{4}{7}
				\Edge{5}{8}
				\Edge{6}{8}
				\Edge{7}{8}	
				
				\rightObjbox{3}{3}{-1}{v}
				\rightObjbox{4}{3}{-1}{x}
				\rightObjbox{7}{3}{-1}{y}
				\rightObjbox{6}{3}{-1}{a}
				
				\end{diagram}}
			
			\put(15,15){\ColorNode{red}}
			\put(15,30){\ColorNode{red}}
			\put(45,15){\ColorNode{blue}}
			\put(30,0){\ColorNode{blue}}
			
			\end{picture}}
	\end{minipage}
	\begin{minipage}{0.22\textwidth}
		\centering
		{\unitlength 0.6mm
			\begin{picture}(60,45)%
			\put(0,0){%
				\begin{diagram}{60}{45}
				
				\Node{1}{30}{0}
				\Node{2}{15}{22.5}
				\Node{3}{30}{15}
				\Node{4}{45}{30}
				\Node{5}{30}{45}						
				\Node{6}{30}{30}

				\Edge{1}{3}
				\Edge{3}{2}
				\Edge{2}{6}
				\Edge{3}{4}
				\Edge{6}{5}
				\Edge{4}{5}
				
				
				\end{diagram}}
			
			\put(15,22.5){\ColorNode{red}}
			\put(30,0){\ColorNode{blue}}
			
			\end{picture}}
\end{minipage}	
\quad\quad\quad\quad
\begin{minipage}{0.22\textwidth}
		\centering
	{\unitlength 0.6mm
		\begin{picture}(60,45)%
		\end{picture}}
	\end{minipage}
\begin{minipage}{0.22\textwidth}
	\centering
	{\unitlength 0.6mm
		\begin{picture}(60,45)%
		\end{picture}}´
\end{minipage}
\caption{For the lattice $\UL$ with a nested interval $\US$ (red),
	the diagrams show all different possibilities to purify the interval relation $\theta_{\US}$ by adding an additional interval $\US_{new}$(blue) with $|\US|=2$.
}
\label{fig:purification}		
\end{figure}
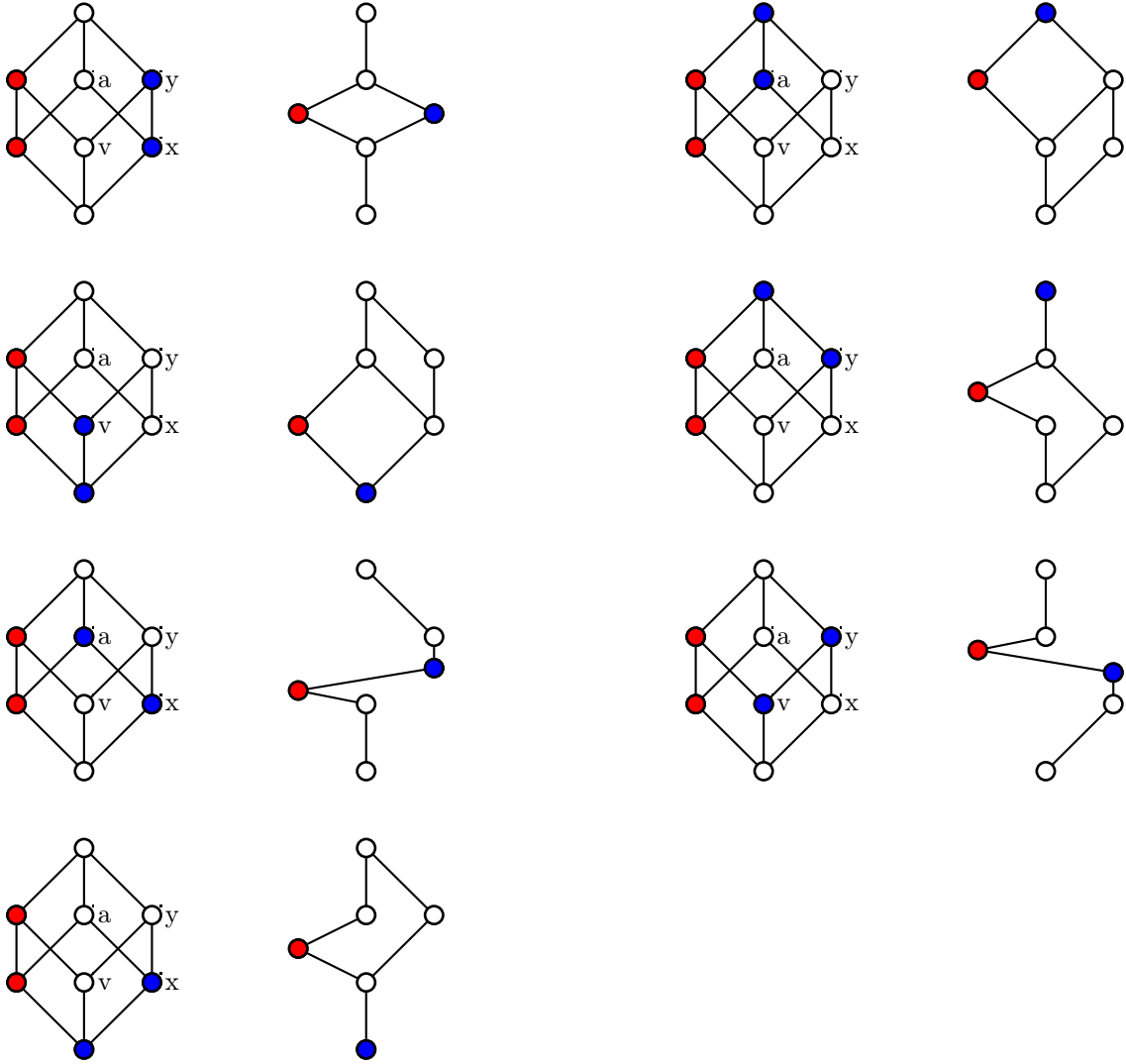

Since our goal was to find a factorization to generate a lattice that can be obtained by a surjective, order-preserving mapping, the lattice-generating interval relation fulfills this purpose.
%
%
%
%
However, the lattice operations $\bigvee$ and $\bigwedge$ are not generally preserved by $\varphi$, i.e., $\varphi$ is, in general, not a lattice homomorphism. 
For example consider the two concepts $a$ and $b$ in~\cref{fig:running_exp_teil2} (left).
Their infimum in the original lattice is $c$.
Nevertheless, in the factor set (right) $[a]\theta\wedge [b]\theta\ne [c]\theta$.
However, it is possible to determine where the lattice operations $\bigvee$ and $\bigwedge$ are not preserved after a factorization using an interval relation:

\begin{lemma}
	Let $\UL$ be a finite lattice and $\theta=\theta_{\US}$ a lattice-generating interval relation on $\UL$.
	Let $u,v,w,x,y,z\in \UL$ with $u\wedge v =w$ and $x\vee y =z$.
	Then:
	\begin{itemize}
		\item[i)] $u\in \US\cup S^{\uparrowtail}, v,w\in S^{\downarrowtail}, v\not= w \Rightarrow [u]\theta\wedge [v]\theta \not = [w]\theta $
		\item[ii)] $x\in \US\cup S^{\downarrowtail}, y,z\in S^{\uparrowtail}, y\not= z \Rightarrow [x]\theta\vee [y]\theta \not = [z]\theta  $
		\item[iii)] $[u]\theta\wedge [v]\theta \not = [w]\theta \Rightarrow u\in \US\cup S^{\uparrowtail}, v\in S^{\downarrowtail}\cup S^{\parallel}, v\not= w$
		\item[iv)] $[x]\theta\vee [y]\theta \not = [z]\theta \Rightarrow x\in \US\cup S^{\downarrowtail}, y\in S^{\uparrowtail}\cup S^{\parallel}, y\not= z$
	\end{itemize}
\end{lemma}
\begin{proof}
	We show i):
	Because $u\in \US\cup S^{\uparrowtail}$ and $v\in  S^{\downarrowtail}$ we have $[v]\theta \le_{\theta} [u]\theta$.
	Since $w<v$ we have $[w]\theta <_{\theta} [v]\theta$.
	This means $[u]\theta\wedge [v]\theta = [v]\theta \not = [w]\theta$.\\
	ii) can be shown analogously.
	
	We show iii):
	We show the contraposition:
	Assumed $v=w$, we have $v\le u$ and therefore $[u]\theta\wedge [v]\theta = [v]\theta = [w]\theta$.
	Thus, let $v\not = w$.
	In case of $u,v\in \US\cup S^{\uparrowtail}$ we have $w\in \US$ or $w\in S^{\uparrowtail}$.
	If $w\in \US$, we have $[u]\theta\wedge[v]\theta=[S]\theta=[w]\theta$,
	if $w\in S^{\uparrowtail}$ the order between the three elements is not affected by the factorization and $[u]\theta\wedge [v]\theta = [w]\theta$ as well.
	In case of $u,v\in S^{\parallel}\cup S^{\downarrowtail}$ we have $w\in S^{\downarrowtail}$ or $w\in S^{\parallel}$.
	The order between the three elements is not affected by the factorization and $[u]\theta\wedge [v]\theta = [w]\theta$.\\
	iv) can be shown analogously.
\end{proof}

As seen in it previous section, crowns play an essential role in determining whether an interval relation is ordered.
Those substructures can also be used to determine the pureness of an interval (relation) as follows: 

\begin{lemma}
	\label{lem:nested=crown}
	Let $\UL$ be a finite lattice and $\US=[S_{\bot},S_{\top}]$ an interval on $\UL$. Then the following equivalence holds:
	\[\US \text{ is nested in } \UL \Leftrightarrow S_{\bot} \text{ and } S_{\top} \text{ are elements of a crown of order $3$ in } \UL\]
\end{lemma}
\begin{proof}
	"$\Leftarrow$":
	Let $A_3\le \UL$ be a crown consisting of 
	$x_1=S_{\bot},x_2,x_3,y_1=S_{\top},y_2$ and $y_3$.
	By definition of a crown we have $y_2\in S^{\uparrowtail}, x_3\in S^{\downarrowtail}$ and $x_2,y_3\in S^{\parallel}$ with $x_2=y_2\wedge y_3, y_3=x_2\vee x_3, y_3\not\le y_2$ and $x_3\not\le x_2$.
	Thus, $\US$ is nested in $\UL$.\\
	"$\Rightarrow$":
	Let $\US$ be nested with the elements $a,v,x,y$.
	Then we have the relations $x\le a$, $S_{\bot}\le a$, $v\le y$, $v\le S_{\top}$, $x\le y$ and $S_{\bot}\le S_{\top}$ as the only relations between those elements.
	Thus the set $s,v,x,y,S_{\bot},S_{\top}$ is a crown of order $3$ in $\UL$.
\end{proof}

Using this, we can now generalize \cref{lem:crown} to lattice-generating interval relations:

\begin{lemma}
	\label{lem:crown2}
	Let $\UL$ be a lattice and $\theta$ an interval relation on $\UL$.
	If $\UL$ does not contain a crown of order $3$ as a suborder, then $\theta$ is lattice-generating.
\end{lemma}

\begin{corollary}
	Let $\UL$ be a lattice and $\theta$ an interval relation on $\UL$.
	If $\UL$ is planar
	$\theta$ is an order-preserving interval relation on $\UL$.
\end{corollary}

Using a lattice-generating interval relation, a lattice arises by factorization so that exactly the chosen intervals of the original lattice implode.
In the following, we investigate this approach on the context side.

\subsection{Context Construction for Interval Factorization}
\label{subsec:kontext}


Since every finite lattice  is isomorphic to a concept lattice $\UBB(\K)$ of a formal context $\K$ and formal contexts tend to be smaller than their corresponding concept lattices,
we will now discuss the corresponding context constructions of our approach.

\begin{definition}
Let $\K=(G,M,I)$ be a formal context and $\{\US_1,\US_2,\dots,\US_k\}$ a set of pairwise disjoint intervals of $\underline{\BB}(\K)$ with $\US_i=[(A_i,B_i),(C_i,D_i)]$.
The incidence relation 
\[I_{\US_1,\dots,\US_k}\coloneqq I\cup\bigcup_{i=1}^k (C_i\times B_i)\] 
is the \emph{enrichment of relation $I$ by the intervals $\US_1,\dots,\US_k$}.
We call the context $\K_{\US}\coloneqq(G,M,I_{\US})$ the \emph{enrichment} of context $\K$ by the interval $\US$.
\end{definition}

Note that simultaneously considering a set of intervals and the iterative enrichment of a relation generally does not end in the same context.
An example is presented in~\cref{fig:iteratives_enrichment_gegenbeispiel}.

%

\begin{figure}[t]
	\centering
	\begin{minipage}{0.49\textwidth}
		\centering
		{\unitlength 0.6mm
			\begin{picture}(60,45)%
			\put(0,0){%
				\begin{diagram}{60}{45}
				
				\Node{1}{30}{0}
				\Node{2}{0}{15}
				\Node{3}{15}{15}
				\Node{4}{30}{15}
				\Node{5}{45}{15}
				\Node{6}{60}{15}
				\Node{7}{0}{30}
				\Node{8}{15}{30}
				\Node{9}{30}{30}
				\Node{10}{45}{30}
				\Node{11}{60}{30}
				\Node{12}{30}{45}
				
				\Edge{1}{2}
				\Edge{1}{3}
				\Edge{1}{4}
				\Edge{1}{5}
				\Edge{1}{6}
				\Edge{2}{7}
				\Edge{2}{8}
				\Edge{3}{8}
				\Edge{3}{9}
				\Edge{4}{9}
				\Edge{4}{10}
				\Edge{5}{10}
				\Edge{5}{11}				
				\Edge{6}{11}
				\Edge{6}{7}
				\Edge{7}{12}
				\Edge{8}{12}
				\Edge{9}{12}
				\Edge{10}{12}
				\Edge{11}{12}
				
				\NoDots\leftObjbox{2}{3}{-1}{6}
				\NoDots\rightObjbox{3}{3}{-1}{7}
				\NoDots\rightObjbox{4}{3}{-1}{8}
				\NoDots\rightObjbox{5}{3}{-1}{9}
				\NoDots\rightObjbox{6}{3}{-1}{10}
				\NoDots\leftObjbox{7}{3}{-1}{1}
				\NoDots\rightObjbox{8}{3}{-1}{2}
				\NoDots\rightObjbox{9}{3}{-1}{3}
				\NoDots\rightObjbox{10}{3}{-1}{4}
				\NoDots\rightObjbox{11}{3}{-1}{5}
				\NoDots\rightObjbox{12}{3}{-3}{11}
				\NoDots\rightObjbox{1}{3}{-1}{12}
				
				\end{diagram}}
			
			\put(30,15){\ColorNode{red}}
			\put(30,30){\ColorNode{red}}
			\put(15,30){\ColorNode{blue}}
			\put(15,15){\ColorNode{blue}}
			
			\end{picture}}			
	\end{minipage}
	\begin{minipage}{0.49\textwidth}
		\centering
		{\unitlength 0.6mm
			\begin{picture}(60,45)%
			\put(0,0){%
				\begin{diagram}{60}{45}
				
				\Node{1}{30}{0}
				\Node{2}{0}{15}
				\Node{3}{12}{19}
				\Node{4}{33}{26}
				\Node{5}{45}{15}
				\Node{6}{60}{15}
				\Node{7}{0}{30}
				\Node{8}{30}{45}
				\Node{9}{60}{30}
				\Node{10}{45}{30}

				\Edge{1}{2}
				\Edge{1}{5}
				\Edge{1}{6}
				\Edge{2}{7}
				\Edge{2}{3}
				\Edge{3}{4}
				\Edge{4}{10}
				\Edge{5}{10}
				\Edge{5}{9}				
				\Edge{6}{9}
				\Edge{6}{7}
				\Edge{7}{8}
				\Edge{10}{8}
				\Edge{9}{8}
				
				\NoDots\leftObjbox{2}{3}{-1}{6}
				\NoDots\rightObjbox{3}{3}{0}{2,7}
				\NoDots\leftObjbox{4}{3}{-4}{3,8}
				\NoDots\rightObjbox{5}{3}{-1}{9}
				\NoDots\rightObjbox{6}{3}{-1}{10}
				\NoDots\leftObjbox{7}{3}{-1}{1}
				\NoDots\rightObjbox{10}{3}{-1}{4}
				\NoDots\rightObjbox{9}{3}{-1}{5}
				\NoDots\rightObjbox{8}{3}{-3}{11}
				\NoDots\rightObjbox{1}{3}{-1}{12}
				
				\end{diagram}}
			
			\put(12,19){\ColorNode{blue}}
			\put(33,26){\ColorNode{red}}
			
			\end{picture}}			
	\end{minipage}\\
	\vspace*{2em}
	\begin{minipage}{0.48\textwidth}
		\centering
		\setlength{\tabcolsep}{4pt}
		\begin{cxt}%
			\att{1}%
			\att{2}%
			\att{3}%
			\att{4}%
			\att{5}%
			\att{6}%
			\att{7}%
			\att{8}%
			\att{9}%
			\att{10}%
			\att{11}%
			\att{12}%
			\obj{x.........x.}{1} %
			\obj{.xB...B...x.}{2} %
			\obj{..xB...B..x.}{3} %
			\obj{...x......x.}{4} %
			\obj{....x.....x.}{5} %
			\obj{xxB..xB...x.}{6} %
			\obj{.xxB..xB..x.}{7} %
			\obj{..xx...x..x.}{8} %
			\obj{...xx...x.x.}{9} %
			\obj{x...x....xx.}{10} %
			\obj{..........x.}{11} %
			\obj{xxxxxxxxxxxx}{12} %
		\end{cxt}
	\end{minipage}
	\begin{minipage}{0.48\textwidth}
		\centering
		\setlength{\tabcolsep}{4pt}
		\begin{cxt}%
			\att{1}%
			\att{2}%
			\att{3}%
			\att{4}%
			\att{5}%
			\att{6}%
			\att{7}%
			\att{8}%
			\att{9}%
			\att{10}%
			\att{11}%
			\att{12}%
			\obj{x.........x.}{1} %
			\obj{.xBB..BB..x.}{2} %
			\obj{..xB...B..x.}{3} %
			\obj{...x......x.}{4} %
			\obj{....x.....x.}{5} %
			\obj{xxBB.xBB..x.}{6} %
			\obj{.xxB..xB..x.}{7} %
			\obj{..xx...x..x.}{8} %
			\obj{...xx...x.x.}{9} %
			\obj{x...x....xx.}{10} %
			\obj{..........x.}{11} %
			\obj{xxxxxxxxxxxx}{12} %
		\end{cxt}
	\end{minipage}
	\caption{A lattice $\UL$ with a two pure intervals $\US_1$, $\US_2$(red and blue highlighted) (top left) and the lattice $\UL/\theta_{\US_1,\US_2}$(top right).
		For $\K=\GMI$, the generic formal context of $\UL$, the enrichments  
		$(G,M,I_{\US_{1},\US_{2}})$ (bottom left) and $(G,M,(I_{\US_{1}})_{\US_{2}})= (G,M,(I_{\US_{2}})_{\US_{1}})$  (bottom right) are given.
		The incidences that are added by the enrichment are depicted by $\bullet$.
	}
	\label{fig:iteratives_enrichment_gegenbeispiel}	
\end{figure}

Therefore, we present the following statements just for single intervals.
We present a one-to-one correspondence between the set of the enrichments of the incidence relation by an interval for a generic formal context $\K$ and the interval relations $\theta_{\US}$ on $\underline{\BB}(\K)$ in the following lemma. 
Note that 
the statement does not hold for reduced formal concepts in general.
This fact is discussed in~\cref{lem:necessary} in more detail.

\begin{lemma}
	\label{lem:one-to-one}
	Let $\UL$ be a lattice and $\K=(G,M,I)$ its generic formal context.
	If $\theta _{\US}$ is an interval relation on $\UL$, then $I_{\US}=I\cup (C\times B)$ is an enrichment of $I$ by the interval $\US=[(A,B),(C,D)]$.
	Conversely, for every enrichment $I_{\US}$ of $I$ by an interval $\US$ the relation $\theta_{\US}$ is an interval relation on $\underline{\BB}(\K)$.
\end{lemma}

\begin{proof}
	Since $\US$ is a single interval, the statement follows directly from the definitions of enrichments and interval relations. 
\end{proof}

In a generic formal context we can also determine weather an interval is pure or nested in the corresponding concept lattice.

\begin{lemma}
	Let $\UL$ be a finite lattice and $\K=\GMI$ its generic formal context.
	Let $\theta_{\US}$ be an interval relation on $\UL$.
	$\US$ is nested interval in $\UL$ if and only if there exists $\mathbb{S}=[H,N]\le \K$ with $\mathbb{S}$ being a Boolean subcontext of dimension $3$, $[\US]^{\theta}\in N$, $[\US]_{\theta}\in H$ and $[\US]_{\theta}I[\US]^{\theta}$.
\end{lemma}
\begin{proof}
	Follows directly from \cref{lem:nested=crown}: 
	A lattice contains a Boolean suborder of dimension $3$ if and only if it contains a crown of order $3$ as a suborder.
	Due to the definition of a generic formal context, there is a Boolean subcontext $[\{a,b,c\},\{x,y,z\}]$ of dimension $3$ in $\K$ precisely if there is a Boolean Suoborder of dimension $3$ and therefore a crown of order $3$ as suborder in the corresponding lattice so that $a,b,c$ are the lower elements of the crown and $x,y,z$ are the upper elements of the crown.
\end{proof}

In~\cref{lem:one-to-one} we considered $\K$ to be generic.
Otherwise, additional reducible
concepts may vanish even if they are not in the chosen interval, as presented in the following example.

\begin{example}
	In~\cref{fig:running_exp_teil3} two contexts that generate (up to isomorphism) the same concept lattice are represented. Both have enrichments by the interval $\US=[(4'',4'),(G,G')]$.
	Consider context $\widetilde{\K}=(\widetilde{G},\widetilde{M},\widetilde{I})$ presented in~\cref{fig:running_exp_teil3} (bottom left) and its corresponding formal context $\UBB(\widetilde{\K})=\underline{\BB}(\K)$ (top right). 
	The enrichment of $\widetilde{I}$ by the red highlighted interval $\US=[(4'',4'),(13'',13')]$ is given by adding the $\bullet$ to $\widetilde{I}$. 
	$\underline{\BB}(\K_{\US})$ is presented in the figure (bottom right).
	It consists of the new generated interval (red) and the remaining concepts in the original order, i.e.
	$\UBB(\widetilde{\K}_{\US})\cong \UBB(\widetilde{\K})/\theta_{\US}$.
	If $\K$ (top middle), the standard context of $\underline{\BB}(\K)$, is considered, the enrichment of the incidence relation by the same interval results in the smaller lattice (top right).
	Since in $\underline{\BB}(\K)$, e.g., the concepts $(5'',5')$ and $(6'',6')$ only differ in an attribute set that is totally included in $\US$, their difference vanishes by the enrichment if no attribute $o$ or $l$ persists to differ them.	
\end{example}

\begin{figure}[tp]
	\centering
	\begin{minipage}{0.35\textwidth}
		\centering
		
		{\unitlength 0.6mm
			
			\begin{picture}(60,80)%
			\put(0,0){%
				
				\begin{diagram}{60}{80}
				\Node{1}{40}{0}
				\Node{2}{15}{15}
				\Node{3}{65}{15}
				\Node{4}{40}{30}
				\Node{5}{28}{55}
				\Node{6}{52}{55}
				\Node{7}{40}{80}
				\Node{8}{3}{40}
				\Node{9}{77}{40}
				\Node{10}{15}{65}
				\Node{11}{65}{65}		
				\Node{12}{27}{40}
				\Node{13}{53}{40}
				\Node{14}{56}{68}
				\Node{15}{60}{73}
				
				\Edge{1}{2}
				\Edge{1}{3}
				\Edge{3}{4}
				\Edge{2}{4}
				\Edge{4}{5}
				\Edge{4}{6}
				\Edge{7}{5}
				\Edge{7}{6}
				\Edge{8}{5}
				\Edge{8}{2}
				\Edge{9}{6}
				\Edge{9}{3}
				\Edge{10}{7}
				\Edge{8}{10}
				\Edge{9}{11}
				\Edge{12}{2}
				\Edge{12}{6}
				\Edge{12}{10}
				\Edge{13}{3}
				\Edge{13}{5}
				\Edge{13}{11}
				\Edge{11}{14}
				\Edge{11}{15}
				\Edge{14}{7}
				\Edge{15}{7}

				\leftObjbox{2}{3}{1}{4}
				\rightObjbox{3}{3}{1}{3}
				\leftObjbox{8}{3}{1}{2}
				\rightObjbox{9}{3}{1}{5}
				\leftObjbox{12}{3}{1}{1}
				\rightObjbox{13}{3}{1}{6}
				\NoDots\leftObjbox{14}{3}{3}{7}
				\rightObjbox{15}{3}{0}{8}
				
				\leftAttbox{10}{3}{1}{a}
				\leftAttbox{5}{3}{1}{c}
				\leftAttbox{6}{3}{1}{b}
				\NoDots\leftAttbox{14}{3}{-2}{d}
				\rightAttbox{15}{3}{1}{e}
				
				{\renewcommand{\AttributeLabelStyle}{\color{blue}}
					\rightAttbox{1}{2}{2}{f}
					\leftAttbox{2}{3}{1}{g}
					\rightAttbox{3}{3}{1}{h}
					\leftAttbox{4}{3}{1}{i}
					\rightAttbox{7}{3}{1}{j}
					\leftAttbox{8}{3}{1}{k}
					\rightAttbox{9}{3}{1}{l}
					\rightAttbox{11}{3}{1}{m}
					\leftAttbox{12}{3}{1}{n}
					\rightAttbox{13}{3}{1}{o}}
				
				{\renewcommand{\ObjectLabelStyle}{\color{blue}}
					\rightObjbox{1}{2}{1}{9}
					\leftObjbox{4}{0}{4}{10}
					\rightObjbox{5}{3}{-2}{11}
					\rightObjbox{6}{2}{-2}{12}
					\leftObjbox{7}{3}{-2}{13}
					\rightObjbox{10}{3}{1}{14}
					\rightObjbox{11}{3}{1}{15}}
				
				\end{diagram}}

			\put(15,65){\ColorNode{red}}		
			\put(15,15){\ColorNode{red}}
			\put(28,55){\ColorNode{red}}
			\put(52,55){\ColorNode{red}}
			\put(40,80){\ColorNode{red}}
			\put(27,40){\ColorNode{red}}
			\put(40,30){\ColorNode{red}}
			\put(3,40){\ColorNode{red}}	
			\end{picture}}			
	\end{minipage}
	\begin{minipage}{0.4\textwidth}
		\centering
		\setlength{\tabcolsep}{4pt}
		\begin{cxt}%
			\att{a}%
			\att{b}%
			\att{c}%
			\att{d}%
			\att{e}%
			\obj{xxB..}{1} %
			\obj{xBx..}{2} %
			\obj{Bxxxx}{3} %
			\obj{xxx..}{4} %
			\obj{BxBxx}{5} %
			\obj{BBxxx}{6} %
			\obj{BBBx.}{7} %
			\obj{BBB.x}{8} %
		\end{cxt}
	\end{minipage}
	\begin{minipage}{0.2\textwidth}
		\centering
		
		{\unitlength 0.6mm
			
			\begin{picture}(40,40)%
			\put(0,0){%
				
				\begin{diagram}{40}{40}
				\Node{1}{0}{40}
				\Node{2}{25}{25}
				\Node{3}{16}{28}
				\Node{4}{20}{33}
				
				\Edge{2}{3}
				\Edge{2}{4}
				\Edge{3}{1}
				\Edge{4}{1}

				\leftObjbox{1}{3}{1}{1,2,4}
				\rightObjbox{2}{3}{1}{3,5,6}
				\NoDots\leftObjbox{3}{3}{3}{7}
				\rightObjbox{4}{3}{0}{8}
				\leftAttbox{1}{3}{1}{a,b,c}
				\NoDots\leftAttbox{3}{3}{-2}{d}
				\rightAttbox{4}{3}{1}{e}
				
				\end{diagram}}
			
			\put(0,40){\ColorNode{red}}	
			
			\end{picture}}			
	\end{minipage}
	%
	%
	%
	%
	%
	%
	%
	%
	%
	%
	%
	%
	%
	%
	%
	%
	%
	\ \\
	\ \\
	\ \\
	\ \\
	\ \\
	\begin{minipage}{0.65\textwidth}
		\centering
		\setlength{\tabcolsep}{4pt}
		\begin{cxt}%
			\att{a}%
			\att{b}%
			\att{c}%
			\att{d}%
			\att{e}%
			\att{f}%
			\att{g}%
			\att{h}%
			\att{i}%
			\att{j}%
			\att{k}%
			\att{l}%
			\att{m}%
			\att{n}%
			\att{o}%
			\obj{xxB...B.BxB..x.}{1} %
			\obj{xBx...B.Bxx..B.}{2} %
			\obj{Bxxxx.BxxxBxxBx}{3} %
			\obj{xxx...x.xxx..x.}{4} %
			\obj{BxBxx.B.BxBxxB.}{5} %
			\obj{BBxxx.B.BxB.xBx}{6} %
			\obj{BBBx..B.BxB..B.}{7} %
			\obj{BBB.x.B.BxB..B.}{8} %
			\obj{xxxxxxxxxxxxxxx}{9} %
			\obj{Bxx...B.xxB..B.}{10} %
			\obj{BBx...B.BxB..B.}{11} %
			\obj{BxB...B.BxB..B.}{12} %
			\obj{BBB...B.BxB..B.}{13} %
			\obj{xBB...B.BxB..B.}{14} %
			\obj{BBBxx.B.BxB.xB.}{15} %
		\end{cxt}
	\end{minipage}
	\begin{minipage}{0.3\textwidth}
		\flushright
		
		{\unitlength 0.6mm
			
			\begin{picture}(50,80)%
			\put(0,0){%
				
				\begin{diagram}{50}{80}
				\Node{1}{10}{0}
				\Node{2}{35}{15}
				\Node{3}{10}{80}
				\Node{4}{47}{40}
				\Node{5}{35}{65}		
				\Node{6}{23}{40}
				\Node{7}{26}{68}
				\Node{8}{30}{73}

				\Edge{1}{2}
				\Edge{4}{2}
				\Edge{4}{5}
				\Edge{6}{2}
				\Edge{6}{5}
				\Edge{5}{7}
				\Edge{5}{8}
				\Edge{7}{3}
				\Edge{8}{3}

				\rightAttbox{1}{2}{2}{f}
				\rightObjbox{1}{2}{1}{9}
				\rightAttbox{2}{3}{1}{h}
				\rightObjbox{2}{3}{1}{3}
				\leftObjbox{3}{3}{-2}{1,2,3,\\10,11,\\12,13,14}
				\rightAttbox{3}{3}{1}{a,b,c,g,i,\\j,k,n}
				\rightAttbox{4}{3}{1}{l}
				\rightObjbox{4}{3}{1}{5}
				\rightAttbox{5}{3}{1}{m}
				\rightObjbox{5}{3}{1}{15}
				\rightAttbox{6}{3}{1}{o}
				\rightObjbox{6}{3}{1}{6}
				\NoDots\leftAttbox{7}{3}{-2}{d}
				\rightAttbox{8}{3}{1}{e}
				\NoDots\leftObjbox{7}{3}{3}{7}
				\rightObjbox{8}{3}{0}{8}
				
				\end{diagram}}
			
			\put(10,80){\ColorNode{red}}	
			
			\end{picture}}			
	\end{minipage}
	\caption{
		A (concept) lattice $\UBB(\K)=\UBB(\widetilde{\K})$ with a pure interval $\US$ highlighted red (top left). The object and attribute labels highlighted in blue are reducible.
		The corresponding reduced formal context $\K=(G,M,I)$ (top middle) and an corresponding
		generic formal context $\widetilde{\K}=(\widetilde{G},\widetilde{M},\widetilde{I})$ (bottom left) have additional incidences marked by $\bullet$, that represent the enrichments of the contexts by $\US$.
		$\UBB(\K_{\US})$ is displayed on the top right, and $\UBB(\widetilde{\K}_{\US})$ on the bottom right.
	}
	\label{fig:running_exp_teil3}
\end{figure}

This illustrates that the lattice, based on the enrichment of an incidence by an interval of a corresponding context, depends on the selection of the context.
It is clear that using the generic formal context leads to an upper bound for the size of the arising lattice, since all concepts are generated by a single object and a single attribute.
In the following, we determine the objects and attributes necessary for generating a lattice isomorphic to the one obtained using the generic formal context. 

\begin{definition}
	Let $\theta_{\US}$ be an interval relation on the lattice $\underline{L}$ with $\US\le\UL$ an interval.
	We call $x\in\UL$ \emph{$\theta$-$\bigvee$-irreducible} if either $x\in J(\UL)$ or for $x\not\in \US$ if $|\{y\in \UL\setminus\US\mid y \text{ is an lower neighbour of } x\}|\le 1$ holds. 
	Analogous
	we call an element $x\in\UL$ \emph{$\theta$-$\bigwedge$-irreducible} if either $x\in M(\UL)$ or for $x\not\in \US$ if $|\{y\in \UL\setminus\US\mid y \text{ is an upper neighbour of } x\}|\le 1$ holds.
\end{definition}
	
\begin{definition}
	Let $\theta_{\US}$ be an interval relation on the lattice $\underline{L}$ with $\US\le\UL$ an interval.
	Let $U=\{x\in \UL\mid x \ is\ \theta-\bigwedge-\text{irreducible}\}$ and $V=\{x\in \UL\mid x \ is\  \theta-\bigvee-irreducible\}$.
	We call a context $\K=(H,N,\le)$ with $V\subseteq H \subseteq L$ and $U\subseteq N \subseteq L$ a \emph{$\theta$-irreducible} context of $\UL$.
\end{definition}

\begin{lemma}
	\label{lem:necessary}
	Let $\UL$ be a lattice, $\theta _{\US}$ an interval relation on $\UL$, $\K=\GMI$ the generic context of $\UL$, and $\K=(H,N,\le)$ a $\theta$-irreducible context of $\UL$.
	Then $\underline{\BB}(H,N,\le)\cong\underline{\BB}(G,M,I_{\US})$ holds.
	
\end{lemma}

\begin{proof}
	We show that every object $g\in G$ with $g\not\in H$ is reducible in $(G,M,I_{\US})$.
	If $g\not\in H$, we have $c=(g'',g')$ is not $\theta$-$\bigvee$-irreducible in $\UL\cong \UBB(\K)$.
	Let $c_1,\dots,c_l\not\in\US$ with $l\ge 2$ be the lower neighbors of $c$ in $\underline{\BB}(\K)$.
	Since the original order relation is preserved by the factorization, $[c_1]\theta,\dots,[c_l]\theta$ are lower neighbors of $[c]\theta$ in $\underline{\BB}(\K)/\theta$.
	Therefore $g$ is reducible in $(G,M,I_{\US})$.
	Analogous $\theta$-$\bigwedge$-reducible elements are unnecessary for the attribute set.
\end{proof}

It follows that not the whole generic context has to be considered in the following but only the context containing all $\theta$-$\bigvee$-irreducible elements as the object set and all $\theta$-$\bigwedge$-irreducible elements as the attribute set.
E.g. the concept $(15'',15')=(m',m'')$ in~\cref{fig:running_exp_teil3} is neither $\theta$-$\bigvee$-irreducible nor $\theta$-$\bigwedge$-irreducible.
Therefore, object $15$ and attribute $m$ have no impact on the factor set.

\begin{lemma}
	\label{lem:lattice-context_zusammenhang}
	Let $\UL$ be a finite lattice, $\K=\GMI$ a $\theta$-irreducible context of $\UL$, and
	$\US\le\UL$ an interval.
	Then:
	\begin{itemize}
		\item[i)] $\US$ is pure $\Leftrightarrow$ $\underline{\BB}(\K_{\US})\cong \UL/\theta_{\US}$
		\item[ii)] $\US$ is nested $\Leftrightarrow$ $\underline{\BB}(\K_{\US}) \not\cong \UL/\theta_{\US}$ and $\underline{\BB}(\K_{\US})$ is the Dedekind–MacNeille completion of $\UL/\theta_{\US}$.
	\end{itemize}
\end{lemma}
\begin{proof}
	i): 
	"$\Rightarrow$:" 
	We assume $\K=(\UL,\UL,\le)$ to be the generic context of $\UL$ and thus $\K_S=(\UL,\UL,\le_{\US})$.
	The factor lattice $\UL/\theta_{\US}$ is isomorphic to the concept lattice of its generic context $(\UL/\theta_{\US},\US/\theta_{\US},\le_{\theta})$.
	Via definition of the order $\le_{\theta}$, for two elements $[g]\theta,[m]\theta\in\UL/\theta$ we have $[x]\theta\le_{\theta}[y]\theta$ if and only if $x\le y$ or $x\le [\US]^{\theta}$ and $y\ge [\US]_{\theta}$ in $\UL$.
	Considering $\K_S$, for two elements $x,y\in \L$ we have $x\le_{\US} y$ if and only if $x\le y$ or $x\in \{c\in\UL\mid c\le [\US]^{\theta}\}$ and $y\in \{c\in\UL\mid c\ge [\US]_{\theta}\}$.
	By identifying each element $x\in \UL$ with the equivalence class $[x]\theta\in\UL/\theta$, the isomorphism between $\UBB(\K_{\US})$ and $\UL/\theta$ follows.
%
	
	"$\Leftarrow$:" If $\underline{\BB}(\K_{\US})\cong \UL/\theta_{\US}$ holds, $\UL/\theta_{\US}$ is a lattice and therefore $\US$ is pure.
	
	ii):
	"$\Rightarrow$:" 
	If $\US$ is nested, $\UL/\theta_S$ is an ordered set but no lattice and thus $\underline{\BB}(\K_{\US}) \not\cong \UL/\theta_S$ holds.
	With~\cite[Theorem 4]{fca-book} follows, that the formal context $(\UL/\theta_S,\UL/\theta_S,\le_{\theta})$ corresponds to the Dedekind–MacNeille completion of $\UL/\theta_S$.
	Further, as seen before, the concept lattices corresponding to $\K_{\US}$  and $(\UL/\theta_S,\UL/\theta_S,\le_{\theta})$ are isomorphic.
	
	"$\Leftarrow$:" 
	If  $\underline{\BB}(\K_{\US})$ is the Dedekind–MacNeille completion of $\UL/\theta_S$, $\underline{\BB}(\K_{\US})$ is isomorphic to the concept lattice $\underline{\BB}((\UL/\theta_{\US},\UL/\theta_{\US},\le))$.
	Since $\underline{\BB}(\K_{\US}) \not\cong \UL/\theta_{\US}$ holds, $\UL/\theta_{\US}$ is not a lattice and therefore $\US$ is nested in $\UL$.
\end{proof}

Therefore we can transfer the statement of $\varphi:\UL\rightarrow\UL/\theta,~x\mapsto[x]\theta$ being surjective and order-preserving to formal contexts:

\begin{lemma}
	Let $\UL$ be a lattice, $\US=[(A_S,B_S),(C_S,D_S)]\le\UL$ a pure interval, and $\K=\GMI$ a $\theta$-irreducible context of $\UL$.
	Then the map
	\begin{align*}
	\varphi:\underline{\BB}(\K)&\rightarrow\underline{\BB}((G,M,I_{\US}))\\
	(A,B)&\mapsto \begin{cases}
	(C_S,B_S) & , (A,B)\in\US \\
	(A,B\cup B_S) & , (A,B)\in S^{\downarrowtail} \\
	(A\cup C_S,B)  & , (A,B)\in S^{\uparrowtail}\\
	(A,B) & , (A,B)\in S^{\parallel}
	\end{cases}
	\end{align*}
	is surjective and order preserving.
\end{lemma}
\begin{proof} 
	follows directly from~\cref{lem:lattice-context_zusammenhang}.
\end{proof}

\section{Discussion and Conclusion}
\label{sec:discussion}

In this work, we presented methods to factorize a lattice so that selected intervals implode.
We started with the investigation of factor lattices generated by complete congruence relations in~\cref{sec:congruenz} and presented an approach to find the finest congruence relation, i.e., the one with as many different congruence classes as possible, to implode a selected interval.
Since every congruence relation is an equivalence relation, the elements of the original lattice can be mapped to the elements of the factor lattice in a unique way.
This property does not hold when using complete tolerance relations, a generalization of the complete congruence relations, instead.
In both cases, the generated factor lattice preserves the original $\bigwedge$- and $\bigvee$-relations.
However, both approaches can result in an over-aggressive reduction of the lattice size, imploding not only the selected interval because these construction results in bigger classes.
To overcome this problem, we introduced another kind of factorization based on newly introduced interval relations in~\cref{sec:new_relation}.
Its equivalence classes include precisely the selected intervals so that
it is possible to implode selected disjoint intervals while preserving all other elements of the original lattice and their order.
As a trade-off, the original $\bigwedge$ and $\bigvee$ operations are no longer preserved in this case.
To ensure that a lattice arises as the factor set, we restrict the approach to single pure intervals.
In this case, by taking advantage of the one-to-one correspondence between interval relations and enrichments of incidence relations by intervals in the corresponding context, we get the corresponding context of the factor set directly.

\bibliographystyle{splncs04}
\bibliography{paper.bib}

\begin{thebibliography}{10}
\providecommand{\url}[1]{\texttt{#1}}
\providecommand{\urlprefix}{URL }
\providecommand{\doi}[1]{https://doi.org/#1}

\bibitem{STUMME2002189}
et~al., G.S.: Computing iceberg concept lattices with titanic. Data \&
  Knowledge Engineering  \textbf{42}(2),  189 -- 222 (2002)

\bibitem{berghammer2012ordnungen}
Berghammer, R.: Ordnungen, Verb{\"a}nde und Relationen mit Anwendungen.
  Springer (2012)

\bibitem{sampling}
Boley, M., G{\"{a}}rtner, T., Grosskreutz, H.: Formal concept sampling for
  counting and threshold-free local pattern mining. In: {SIAM} International
  Conference on Data Mining, ({SDM} 2010). pp. 177--188. {SIAM} (2010)

\bibitem{czedli1982factor}
Cz{\'e}dli, G.: Factor lattices by tolerances. Acta Sci. Math.(Szeged)
  \textbf{44}(1-2),  35--42 (1982)

\bibitem{diasreducing}
Dias, S., Vieira, N.: Reducing the size of concept lattices: The {JBOS}
  approach. In: 7th International Conference on Concept Lattices and Their
  Applications (CLA 2010). {CEUR} Workshop Proceedings, vol.~672, pp. 80--91
  (2010)

\bibitem{Durrschnabel}
D{\"u}rrschnabel, D., Koyda, M., Stumme, G.: Attribute selection using
  contranominal scales. In: Graph-Based Representation and Reasoning. pp.
  127--141. Springer International Publishing, Cham (2021)

\bibitem{fca-book}
Ganter, B., Wille, R.: Formal Concept Analysis - Mathematical Foundations.
  Springer (1999)

\bibitem{hanika2019relevant}
Hanika, T., Koyda, M., Stumme, G.: Relevant attributes in formal contexts. In:
  24th International Conference on Conceptual Structures, ({ICCS} 2019).
  Lecture Notes in Computer Science, vol. 11530, pp. 102--116. Springer (2019)

\bibitem{kelly1974crowns}
Kelly, D., Rival, I.: Crowns, fences, and dismantlable lattices. Canadian
  Journal of Mathematics  \textbf{26}(5),  1257--1271 (1974).
  \doi{10.4153/CJM-1974-120-2}

\bibitem{kolibiar1987congruence}
Kolibiar, M.: Congruence relations and direct decompositions of ordered sets.
  Acta Sci. Math.(Szeged)  \textbf{51}(1-2),  129--135 (1987)

\bibitem{Kuitche2018}
Kuitch{\'{e}}, R., Temgoua, R., Kwuida, L.: A similarity measure to generalize
  attributes. In: 14th International Conference on Concept Lattices and Their
  Applications ({CLA} 2018). {CEUR} Workshop Proceedings, vol.~2123, pp.
  141--152 (2018)

\bibitem{Kumar}
Kumar, C.: Knowledge discovery in data using formal concept analysis and random
  projections. International Journal of Applied Mathematics and Computer
  Science  \textbf{21}(4),  745--756 (2011)

\bibitem{kuzuetsov1990stability}
Kuznetsov, S.: Stability as an estimate of the degree of substantiation of
  hypotheses derivedon the basis of operational similarity.
  Nauchno-Tekhnicheskaya Informatsiya, Seriya 2  (1990)

\bibitem{ordercongruence}
Moorthy, C., Karpagavalli, S.: A congruence relation in partially ordered sets
  \textbf{13},  401--405 (08 2015)

\bibitem{qi2019multi}
Qi, J., Wei, L., Wan, Q.: Multi-level granularity in formal concept analysis.
  Granular Computing  \textbf{4}(3),  351--362 (2019)

\bibitem{snavsel2009congruences}
Sn{\'a}{\v{s}}el, V., Jukl, M.: Congruences in ordered sets and lu compatible
  equivalences. Acta Universitatis Palackianae Olomucensis. Facultas Rerum
  Naturalium. Mathematica  \textbf{48}(1),  153--156 (2009)

\bibitem{zadeh1997toward}
Zadeh, L.A.: Toward a theory of fuzzy information granulation and its
  centrality in human reasoning and fuzzy logic. Fuzzy Sets Syst.
  \textbf{90}(2),  111--127 (1997)

\end{thebibliography}
\end{document}